\documentclass{article}
\usepackage{graphicx} 
\usepackage{amsmath, amsfonts, amsthm, xcolor, comment}
\newtheorem{theorem}{Theorem}
\newtheorem{corollary}{Corollary}
\usepackage[margin=1in]{geometry}
\usepackage[colorinlistoftodos]{todonotes}
\usepackage{subcaption}
\usepackage{float} 
\usepackage[colorlinks=true, citecolor=blue, linkcolor=blue, urlcolor=blue]{hyperref}
\usepackage[square,numbers,sort&compress]{natbib}
\usepackage{authblk}

\title{Variational Quantum Homotopy Perturbation Method to Solve Nonlinear Partial Differential Equations}
\author[1]{Jungin E. Kim} 
\author[1]{Eunsik Choi}
\author[1*]{Yan Wang}

\affil[1]{George W. Woodruff School of Mechanical Engineering, Georgia Institute of Technology, 801 Ferst Dr NW, Atlanta, GA 30318}
\affil[ ]{*Corresponding author: yan-wang@gatech.edu}

\date{\today}

\begin{document}

\maketitle

\section*{Abstract}
Solving nonlinear partial differential equations is important in various scientific and engineering applications. Recently, quantum computing was introduced as an alternative computational paradigm for solving nonlinear partial differential equations. In this paper, a new method called the variational quantum homotopy perturbation method is proposed to improve the scalability of solving nonlinear partial differential equations through two aspects. First, the dimension of the Hilbert space remains the same after the nonlinear partial differential equation is linearized through the homotopy perturbation.  Second, the solutions are obtained with a variational quantum simulation framework, where the number of qubits is decreased with functional encoding and the depth of parametrized circuits is reduced. The additional contribution of this paper is the introduction of new criteria for selecting the homotopy series truncation order and circuit depth for cost-effective variational quantum homotopy perturbation method.
The proposed approach is demonstrated with several examples, including the vorticity transport equation and the reduced magnetohydrodynamics equations.

\vspace{3pt}

\noindent\textbf{Keywords:} quantum scientific computing, nonlinear differential equations, homotopy perturbation method, variational quantum simulation

\section{Introduction}

Solving nonlinear partial differential equations (PDEs) is important in various science and engineering applications such as fluid dynamics, vibration analysis, magnetohydrodynamics, and phase transitions.  
The nonlinearity in the PDEs makes it challenging to find the analytical forms of solutions. Numerical methods are mostly required. 
The difficulty of solving a nonlinear PDE can be reduced by converting the original problem to a linear one. 
In the discretization approach such as finite-difference method (FDM), the PDE is numerically approximated as a system of linear equations.
Increasing the resolution of discretizations improves the accuracy of the approximated solution. Some problems such as turbulence flow simulation require very high resolution. 
Solving nonlinear PDEs faces the major scalability challenge. 

Recently, quantum computers have been utilized in various ways to solve nonlinear PDEs. To be compatible with the quantum computing operations, the original nonlinear PDEs are either discretized into linear equations or converted to linear ordinary differential equations (ODEs). 
The converted linear equations can be solved with quantum linear equation solvers such as the Harrow-Hassidim-Lloyd algorithm \cite{Harrow2009}, matrix inversion based on linear combination of unitaries \cite{childs2017quantum}, quantum singular value transformation \cite{gilyen2019quantum, martyn2021grand}, and variational quantum linear solver \cite{bravo2019variational}. 
Although quantum linear equation solvers show the promise of exponentially reduced spatial complexity in comparison with classical methods, such advantage diminishes as the condition number increases in real-world applications.   Several quantum algorithms rely on linearization approaches such as Koopman-von Neumann linearization \cite{joseph2020koopman} and Carleman linearization \cite{liu2021efficient}. In these linearization approaches, a nonlinear PDE defined in a finite-dimensional state space is transformed into linear PDEs defined in an infinite-dimensional Hilbert space.
The linear PDEs are solved with quantum algorithms such as linear combination of unitaries \cite{berry2017quantum, childs2021high}, qubitization \cite{low2019hamiltonian}, Schr\"odingerization \cite{jin2022quantumschrodingerization}, moment-matching dilation \cite{li2026linear}, and variational quantum algorithms \cite{liu2021variational}. Hamiltonian simulation methods that are targeted for fault-tolerant quantum computing require deep circuits. Variational quantum algorithms for near-term computers usually require a large number of parameters to be optimized, which suffers from the barren plateau issue.  A third approach to linearize nonlinear PDEs includes the homotopy perturbation \cite{xue2021quantum} and homotopy analysis methods \cite{xue2025quantum, choi2026lindbladian}, where the solution is approximated as a power series so that the nonlinear problem is converted to a system of recursive linear PDEs known as linear deformation equations. 
A fourth approach to solve nonlinear PDEs specifically for fluid dynamics simulation is quantum lattice Boltzmann method \cite{yepez2001quantum,yepez2002quantum,succi2015quantum,itani2024quantum}, where the original Navier-Stokes equations is reformulated as the lattice Boltzmann equation and the dynamics of quantum states models the evolution of particle densities.

In this paper, a new nonlinear PDE solver called the variational quantum homotopy perturbation method (VQHPM) is proposed. The VQHPM improves the scalability of solving nonlinear PDEs through two aspects. First, the dimension of the Hilbert space remains the same after the nonlinear PDE is linearized through the homotopy perturbation.  Second, the solutions are obtained with a recently developed variational quantum simulation (VQS) framework \cite{sul2024quantum, sul2025generic}, where the number of qubits is decreased with functional encoding and the depth of parametrized circuits is reduced.  With scalability improvements, VQHPM is practical for solving nonlinear PDEs on current quantum computers. The additional contribution of this paper is the introduction of new criteria for selecting the homotopy series truncation order and circuit depth for cost-effective VQHPM.  The first criterion is based on the rigorous analysis outcome that the required truncation order increases logarithmically as the targeted approximation error decreases. The second one is based on a new theorem that the circuit depth increases logarithmically with respect to the determinant of the Fubini-Study metric.

The remainder of the paper is structured as follows. An overview of existing quantum nonlinear differential equation solvers is provided in Section \ref{sec:quantumnonlindesolvers}. The proposed VQHPM is described in Section \ref{sec:proposedqhpm}. Details about the VQS framework for solving the linear deformation equations are provided in Section \ref{sec:VQS}. In Section \ref{sec:results}, VQHPM is demonstrated with two examples of nonlinear PDEs, including the vorticity transport equation and reduced magnetohydrodynamics. Conclusions and future work are discussed in Section \ref{sec:conclusions}.

\section{Existing Work of Quantum Nonlinear Differential Equation Solvers} \label{sec:quantumnonlindesolvers} 

Several quantum algorithms have been developed to solve nonlinear differential equations.  For instance, some algorithms involve circuits to prepare multiple copies of quantum states to encode nonlinearities.  Leyton and Osborne \cite{leyton2008quantum} applied Hamiltonian simulations to evolve two copies of quantum states simultaneously.  In the algorithm of Lloyd et al. \cite{lloyd2020quantum}, multiple copies of quantum states are prepared and nonlinear PDEs are linearized before quantum linear equation solvers are applied. Lubasch et al. \cite{lubasch2020variational} developed a variational quantum algorithm where a parameterized circuit consists of a quantum nonlinear processing unit to compute the nonlinear terms in the nonlinear Schr\"odinger's equation.  The processing unit consists of controlled NOT operations which enable bit-wise multiplication of multiple copies of variables encoded with basis states.  This algorithm was utilized by Sarma et al. \cite{sarma2024quantum} to solve other nonlinear PDEs in fluid dynamics, biological processes, and finance. Variational quantum algorithms can also handle nonlinearities without multiple copies of variables.  Jaksch et al. \cite{jaksch2023variational} devised a parameterized circuit which calculates the PDE residual as an expectation value. 

In some algorithms, nonlinear differential equations are solved recursively from previous time steps. Gaitan \cite{gaitan2020finding} proposed an algorithm to solve the Navier-Stokes equations, where the quantum amplitude estimation is applied to estimate the time average of solutions under nonlinear ODE operators. This algorithm has also been applied to solve Burgers' equation \cite{oz2022solving} and the radiation diffusion equation \cite{gaitan2024simulating}. Shukla and Vedula \cite{shukla2023hybrid} proposed another algorithm where the solution is expanded as a linear combination of Walsh-Hadamard basis functions. The coefficients of the time integration matrix are obtained by performing the Walsh-Hadamard transform, which is easier to implement than quantum amplitude estimation.  In some other algorithms, the time integration is formulated as a system of linear equations, where the solution of the next time step is recursively obtained from the previous time steps.  An algorithm based on this approach was devised by Gnanasekaran et al. \cite{gnanasekaran2023efficient}, where the discretized Fokker-Planck equation is solved with matrix inversion based on a linear combination of unitaries \cite{berry2017quantum}. Variational quantum linear equation solvers have also been utilized to solve the Lorenz system \cite{fathi2024hybrid} and reservoir flow equations \cite{rao2024performance}.

Other quantum nonlinear differential equation solvers involve linearizing the original nonlinear problem.  One linearization approach is Koopman-von Neumann linearization \cite{joseph2020koopman}, where the nonlinear time evolution of a dynamical system is reformulated as a linear time evolution of an observable of the system. Quantum algorithms based on this approach were developed by Jin et al. \cite{jin2023time} to solve nonlinear ODEs and nonlinear Hamilton-Jacobi PDEs. The discretized linear differential equation is solved with Hamiltonian simulation or quantum linear equation solvers. Koopman-von Neumann linearization in combination with quantum singular value transformation has also been applied to simulate plasma dynamics \cite{higuchi2025quantum}.

Another linearization approach is Carleman linearization, where a finite-dimensional nonlinear differential equation is transformed into an infinite-dimensional system of linear differential equations by introducing monomials as additional state variables.  Liu et al. \cite{liu2021efficient} developed a method based on Carleman linearization for nonlinear dissipative ODEs. The linearized ODEs are discretized over time and solved with matrix inversion based on a linear combination of unitaries.  The algorithm was further extended by Krovi \cite{krovi2023improved} to linear ODEs with non-diagonalizable matrices, where bounds for the matrix norm are established which ensure that the discretized linear system is stable.  Carleman linearization has also been applied to solve nonlinear PDEs, which are discretized over space prior to being converted to linear ODEs.  This method was used to solve the advection-diffusion \cite{demirdjian2022variational} and reaction-diffusion equations \cite{liu2023efficient}. To simulate fluid flow, Itani et al. \cite{itani2024quantum} proposed a quantum lattice Boltzmann method where Carleman linearization is applied to derive linear bosonic operators which embed nonlinear collisions. Several improvements to Carleman linearization have also been proposed. For instance, Wu et al. \cite{wu2025quantum} generalized Carleman linearization for non-dissipative nonlinear PDEs by enforcing a no-resonance condition, where any eigenvalue of the discretized linear system cannot be a linear combination of other eigenvalues.  To reduce computational cost, Costa et al. \cite{costa2025further} proposed a rescaling strategy where the condition number of the discretized linear system is reduced.

A third linearization approach includes homotopy methods, where the solution is expanded as a power series which consists of an initial guess term and nonlinear correction terms. This power series expansion is utilized to transform the original nonlinear problem into a recursive sequence of linear differential equations known as linear deformation equations. In contrast to the first two linearization approaches, homotopy methods exhibit improved scalability through which the Hilbert space dimension remains constant. Xue et al. \cite{xue2021quantum} proposed a homotopy perturbation method for nonlinear dissipative ODEs, where the linear deformation equations are combined into a larger system of linear ODEs. The solution to the combined system of ODEs is obtained by performing matrix inversion based on a linear combination of unitaries. Xue et al. \cite{xue2025quantum} also proposed a quantum homotopy analysis method for nonlinear PDEs. The linear deformation equations are combined into a single system of linear PDEs through quantum-compatible linearization, where the nonlinear terms are defined as additional state variables. Recently, Choi et al. \cite{choi2026lindbladian} devised a Lindbladian homotopy analysis method to simulate non-unitary and nonlinear dynamics. In this method, the linear deformation equations are reformulated as one homogeneous autonomous system of ODEs by coupling solutions from different homotopy orders. The non-unitary time evolution of the autonomous system is simulated with Lindbladian dynamics by embedding the linear dissipative operator into the jump operator of the Lindblad master equation.

\section{Proposed Variational Quantum Homotopy Pertubation Method} \label{sec:proposedqhpm}
\subsection{Linearization of Nonlinear PDE} \label{subsc:linearization}
The VQHPM is devised to solve a nonlinear PDE defined as
\begin{equation} \label{eq:nonlinearPDE}
    \frac{\partial {u}(\boldsymbol{r}, t)}{\partial t} =
    \mathcal{L}({u}(\boldsymbol{r}, t)) + \mathcal{N}({u}(\boldsymbol{r}, t)),
\end{equation}
where ${u}(\boldsymbol{r}, t)$ is the solution at spatial location $\boldsymbol{r}$ and time $t$. On the right side of Eq. (\ref{eq:nonlinearPDE}), $\mathcal{L}$ and $\mathcal{N}$ are linear and nonlinear differential operators, respectively.  The initial condition is denoted as ${u}(\boldsymbol{r}, t=0) = {u}_0$.

Eq. (\ref{eq:nonlinearPDE}) is transformed into a system of linear deformation equations through the homotopy perturbation method, which is based on a continuous transformation from a simple linear PDE to the original nonlinear PDE. This transformation is formulated as
\begin{equation} \label{eq:zerothorderequation}
    (1-p)\mathcal{L}_h\left({v}(\boldsymbol{r}, t) - {v}^{(0)}(\boldsymbol{r}, t)\right) + p\mathcal{N}_h({v}(\boldsymbol{r}, t)) = 0,
\end{equation}
where $p \in [0, 1]$ is an embedding parameter and ${v}^{(0)}$ is an initial guess for the nonlinear PDE solution $u$. Eq. (\ref{eq:zerothorderequation}) consists of linear differential operator
\begin{equation} \label{eq:qhpmlinearoperator}
    \mathcal{L}_h\left({v}(\boldsymbol{r}, t) - {v}^{(0)}(\boldsymbol{r}, t)\right) =  \frac{\partial {v}(\boldsymbol{r}, t)}{\partial t} - \frac{\partial {v}^{(0)}(\boldsymbol{r}, t)}{\partial t} - 
    \mathcal{L}({v}(\boldsymbol{r}, t)) + \mathcal{L}\left({v}^{(0)}(\boldsymbol{r}, t)\right)
\end{equation}
and nonlinear operator
\begin{equation} 
\label{eq:qhpmnonlinearoperator}
    \mathcal{N}_h({v}(\boldsymbol{r}, t)) = \frac{\partial {v}(\boldsymbol{r}, t)}{\partial t} -
    \mathcal{L}({v}(\boldsymbol{r}, t)) - \mathcal{N}({v}(\boldsymbol{r}, t)).
\end{equation}
In Eq. (\ref{eq:zerothorderequation}), both $\mathcal{L}_h$ and $\mathcal{N}_h$ are applied to the approximated solution
\begin{equation} \label{eq:homotopyseries}
    {v}(\boldsymbol{r}, t) = {v}^{(0)}(\boldsymbol{r}, t) + \sum_{j=1}^m p^j {v}^{(j)}(\boldsymbol{r}, t),
\end{equation} where $m$ is the homotopy order and ${v}^{(j)}$ is the $j$\textsuperscript{th}-order nonlinear correction term. The linear operator $\mathcal{L}_h$ is also applied to the initial guess ${v}^{(0)}$.

Based on Eq. (\ref{eq:homotopyseries}), Eq. (\ref{eq:zerothorderequation}) is expanded as an $m$-degree polynomial with respect to $p$.  A collection of linear deformation equations is subsequently obtained by separating the polynomial coefficients in the expanded Eq. (\ref{eq:zerothorderequation}). The linear deformation equations are derived as
\begin{equation} \label{eq:recursivePDEs}
\begin{gathered}
    \mathcal{L}_h\left({v}^{(0)}(\boldsymbol{r}, t)\right)  = \mathcal{L}_h\left({v}^{(0)}(\boldsymbol{r}, t)\right) \\
    \mathcal{L}_h\left({v}^{(1)}(\boldsymbol{r}, t)\right) = \mathcal{R}_1\left({v}^{(0)}(\boldsymbol{r}, t)\right) \\
     \mathcal{L}_h\left({v}^{(2)}(\boldsymbol{r}, t)\right) = \mathcal{R}_2\left({v}^{(0)}(\boldsymbol{r}, t), {v}^{(1)}(\boldsymbol{r}, t)\right) \\
    \vdots \\
    \mathcal{L}_h\left({v}^{(m)}(\boldsymbol{r}, t)\right) = \mathcal{R}_m\left({v}^{(0)}(\boldsymbol{r}, t), {v}^{(1)}(\boldsymbol{r}, t), \dots, {v}^{(m-1)}(\boldsymbol{r}, t)\right), \\
\end{gathered}
\end{equation}
where $\mathcal{R}_j\left({v}^{(0)}, {v}^{(1)}, \dots, {v}^{(j-1)}\right)$ in the $j$\textsuperscript{th}-order deformation equation is obtained by evaluating all nonlinear correction terms up to the $(j-1)$\textsuperscript{th} order at each time step. That is, the initial guess ${v}^{(0)}$ and nonlinear correction terms ${v}^{(1)}$, \dots, ${v}^{(m)}$ are obtained by recursively solving the deformation equations with the increasing order. 
The approximated  solution ${v}$ is subsequently computed based on Eq. (\ref{eq:homotopyseries}), where $p$ is set to 1.

\subsection{Decomposition of Linear Deformation Equations}
The linear deformation equations in Eq. (\ref{eq:recursivePDEs}) are nonhomogeneous. The difficulty of solving the $j$\textsuperscript{th}-order linear deformation equation is reduced by decomposing the original problem into two sub-problems of computing the particular and homogeneous components of ${v}^{(j)}$.

The time-independent particular component, which is denoted as $\hat{v}^{(j)}(\boldsymbol{r})$, represents the steady-state behavior of the $j$\textsuperscript{th}-order linear deformation equation. The particular component is obtained from a steady-state PDE defined as
\begin{equation} \label{eq:steadystatepde} 
\mathcal{L}_h\left(\hat{{v}}^{(j)}\left(\boldsymbol{r}\right)\right) =\mathcal{R}_j\left(\hat{{v}}^{(0)}\left(\boldsymbol{r}\right), \hat{{v}}^{(1)}\left(\boldsymbol{r}\right), ..., \hat{{v}}^{(j-1)}\left(\boldsymbol{r}\right)\right),
\end{equation}
where  $\mathcal{R}_j\left(\hat{{v}}^{(0)}, \hat{{v}}^{(1)}, ..., \hat{{v}}^{(j-1)}\right)$ depends on the previously computed steady-state solutions at lower orders. The steady-state solution $\hat{{v}}^{(0)}$ is
equivalent to the steady-state solution of the nonlinear PDE in Eq. (\ref{eq:nonlinearPDE}). The higher-order steady-state solutions $\hat{{v}}^{(1)}$, \dots, $\hat{{v}}^{(j)}$ are computed recursively with the increasing order. That is, after $\hat{{v}}^{(0)}$, $\hat{{v}}^{(1)}$, \dots, $\hat{{v}}^{(j-1)}$ are calculated, $\hat{{v}}^{(j)}$ is obtained by solving the $j$\textsuperscript{th}-order steady-state PDE.  Different methods can be used to solve the steady-state linear PDEs. For instance, they can be converted to systems of linear equations by discretizing the spatial domain as in FDM.

The time-dependent homogeneous component, which is denoted as $\Tilde{{v}}^{(j)}(\boldsymbol{r}, t)$, is obtained by solving
\begin{equation} \label{eq:homogeneouspde}
     \mathcal{L}_h\left(\Tilde{{v}}^{(j)}(\boldsymbol{r}, t)\right) = 0.
\end{equation}
This PDE is similar to the $j$\textsuperscript{th}-order linear deformation equation in Eq. (\ref{eq:recursivePDEs}), with the exception of $\mathcal{R}_j$ being omitted. Eq. (\ref{eq:homogeneouspde}) is solved with the VQS framework, which is described in Section \ref{sec:VQS}.

After the solutions of the decomposed linear deformation equations are obtained, the solution is approximated as
\begin{equation} \label{eq:homotopyseriesdecomposed}
    {v}(\boldsymbol{r}, t) = {v}^{(0)}(\boldsymbol{r}, t) + \sum_{j=1}^m p^j \left[ \hat{{v}}^{(j)}(\boldsymbol{r}) + \Tilde{{v}}^{(j)}(\boldsymbol{r}, t) \right].
\end{equation}
As the nonlinearity increases, more correction terms are needed to improve the accuracy of the approximated solution $v$. Nevertheless, as $m$ increases, VQHPM can become computationally expensive because a large number of linear deformation equations in Eq. (\ref{eq:recursivePDEs}) must be solved. Therefore, the minimal homotopy order that allows us to achieve the targeted threshold of approximation error needs to be estimated.

\subsection{Homotopy Order Selection Criterion} \label{subsec:homotopytruncation}

The relationship between the homotopy order $m$ and the approximation error $\left\Vert {u} - \sum_{j=1}^{m+1} {v}^{(j-1)}\right\Vert$ has been studied. The approximation error is reduced with the convergence rate in the polynomial order of contractive ratio.
\begin{theorem} \label{thm:solutionapproximation}
\cite{odibat2010study} 
Let $m$ denote the homotopy order and $q = \max_{j \in [1, m]} \left\Vert v^{(j)} \right\Vert / \left\Vert v^{(j-1)} \right\Vert$ denote the contractive ratio. Then
\begin{equation}
\label{eq:qratio}
\left\Vert  {u} - \sum_{j=1}^{m+1} {v}^{(j-1)}\right\Vert \leq \frac{q^{m+1}}{1-q} \left\Vert {v}^{(0)}\right\Vert.
\end{equation}
\end{theorem}

\begin{corollary}
\label{thm:truncation}
Given the targeted approximation error $\epsilon$, the homotopy order $m$ needs to satisfy
\begin{equation} 
    m \geq \max\left(0, \frac{\log \left( \frac{\epsilon(1-q)}{\Vert {v}^{(0)} \Vert} \right)}{\log\left( q \right)} - 1\right).
\end{equation}
\begin{proof}
It is sufficient to assume from Eq. (\ref{eq:qratio}) that
\begin{equation} \label{eq:lessthanepsilon}
    \frac{q^{m+1}}{1-q} \left\Vert {v}^{(0)}\right\Vert \leq \epsilon
\end{equation}
so that $\left\Vert {u} - \sum_{k=1}^{m+1} {v}^{(k-1)}\right\Vert \leq \epsilon$.  The re-arrangement of Eq. (\ref{eq:lessthanepsilon}) yields
\begin{equation}
\label{eq:lowerbound}
    m \geq \frac{\log\left(\frac{\epsilon(1-q)}{\left\Vert {v}^{(0)} \right\Vert}\right)}{\log(q)} - 1.
\end{equation}
The right-hand side of Eq. (\ref{eq:lowerbound}) is negative when $\Vert {v}^{(0)} \Vert < \epsilon{(1-q)}/{q}$. Since $q \in (0, 1)$, $\Vert {v}^{(0)} \Vert < \epsilon$, which means that ${v}^{(0)}$ is accurate enough to meet the target approximation error. Therefore,
\begin{equation}
    m \geq \text{max}\left(0, \frac{\log\left(\frac{\epsilon(1-q)}{\left\Vert {v}^{(0)} \right\Vert}\right)}{\log(q)} - 1\right).
\end{equation}
\end{proof}
\end{corollary}

\noindent Corollary \ref{thm:truncation} provides the guidance to select the minimum value of $m$ given the targeted approximation error. As the value of $\epsilon$ decreases, the minimum number of nonlinear correction terms increases logarithmically. Because $q \in (0, 1)$, the complexity of the homotopy order is in the order $\mathcal{O}(\log(\epsilon^{-1}))$.

\section{Variational Quantum Simulation Framework} 
\label{sec:VQS}

The homogeneous components $\Tilde{v}^{(j)}(\boldsymbol{r}, t)$'s are obtained by solving Eq. (\ref{eq:homogeneouspde}) with our recently developed VQS framework \cite{sul2024quantum, sul2025generic}. In this framework, $\Tilde{v}^{(j)}$ is expanded as a linear combination of interpolation functions through quantum functional encoding.  The parameters of the VQS circuit to obtain $\Tilde{v}^{(j)}$ are evolved over time with Euler's method.  The time derivatives of circuit parameters are obtained by solving a system of linear equations. This linear system is a reformulation of McLachlan's variational principle, which minimizes the residual of Eq. (\ref{eq:homogeneouspde}).

\subsection{Quantum Functional Encoding} \label{subsec:qfe}
In quantum functional encoding, $\Tilde{{v}}^{(j)}$ is expanded as a linear combination of basis functions. In VQHPM, $\Tilde{{v}}^{(j)}$ is expanded as
\begin{equation} \label{eq:collocation}
    \Tilde{{v}}^{(j)}(\boldsymbol{r}, t) \approx \sum_{k=1}^{d} \Tilde{v}^{(j)}(\boldsymbol{r}_k, t)\phi_k(\boldsymbol{r}),
\end{equation}
where $d$ is the total number of collocation points, $\boldsymbol{r}_k$ is the location of the $k$\textsuperscript{th} collocation point, and $\phi_k$ is an interpolation function.  The quantum state that encodes the expansion coefficients in Eq. (\ref{eq:collocation}) is
\begin{equation}
\label{eq:quantumfunctionalexpansion}
    |\Tilde{v}^{(j)}(\boldsymbol{r}, t)\rangle = \sum_{k=1}^{d} \frac{\Tilde{v}^{(j)}(\boldsymbol{r}_k, t)}{\alpha}|k\rangle,
\end{equation}
where $\alpha$ is a scale factor and $|k\rangle$ is a computational basis state corresponding to the $k$\textsuperscript{th} collocation point.  The state $|\Tilde{v}^{(j)}(\boldsymbol{r}, t)\rangle$ is obtained with a parameterized circuit consisting of $n = \log_2 d$ qubits.  

The circuit is constructed with operator
\begin{equation} \label{eq:vqsansatz}
    U(\boldsymbol{\theta}) = \prod_{l=1}^w  \left( U_e  U_r(\boldsymbol{\theta}^{(w+1-l)}) \right),
\end{equation}
where $w$ is the circuit depth and $\boldsymbol{\theta}^{(l)}$ is a vector which includes $n$ circuit parameters $\theta_{(l-1)n+1}, \dots, \theta_{ln}$.  The $l$\textsuperscript{th}  layer of rotation gates is defined as $U_r(\boldsymbol{\theta}^{(l)}) = \otimes_{s = 1}^n \exp(-i
\theta_{(l-1)n+s}\sigma_Y/2)$, where $\sigma_Y$ is the Pauli-Y matrix.  Each layer of $U_r$'s is constructed with $n$ RY gates that allow the circuit to explore the Hilbert space.  The $l$\textsuperscript{th} layer of entanglement operators is defined as $U_e = \prod_{s=1}^n \left( |0^{(s)}\rangle\langle 0^{(s)}| \otimes I_2^{(s+1)} + |1^{(s)}\rangle\langle 1^{(s)}| \otimes \sigma_X^{(s+1)} \right)$, where $I_2$ is the $2 \times 2$ identity matrix and $\sigma_X$ is the Pauli-X matrix.  The circuit alternates between $U_r$ and $U_e$ for $w$ repetitions.  The resulting state $|\Tilde{v}^{(j)}\rangle$, which can be obtained with quantum state tomography, is multiplied by $\alpha$ to obtain the homogeneous component $\Tilde{v}^{(j)}$.

\subsection{Variational Quantum Simulation} \label{subsec:maclachlan}

In our VQS framework \cite{sul2024quantum, sul2025generic}, the values of $\boldsymbol{\theta}(t)$ and $\alpha(t)$ which result in the state $|\Tilde{{v}}^{(j)}(\boldsymbol{r}, t)\rangle$ are obtained by solving McLachlan's variational principle. Based on the homogeneous PDE in Eq. (\ref{eq:homogeneouspde}), McLachlan's variational principle is defined as
\begin{equation} \label{eq:mclachlan}
    \delta \left\Vert \frac{\partial}{\partial t}|\Tilde{v}^{(j)}(\boldsymbol{\Theta}(t))\rangle - {L}_h |\Tilde{v}^{(j)}(\boldsymbol{\Theta}(t))\rangle \right\Vert = 0,
\end{equation}
where $\boldsymbol{\Theta}(t) = \begin{bmatrix} \alpha(t) & \boldsymbol{\theta}(t) \\ 
\end{bmatrix}$ and $L_h$ is the matrix representation of $\mathcal{L}_h$. Eq. (\ref{eq:mclachlan}) is subsequently reformulated as a system of $wn+1$ linear equations
\begin{equation} \label{eq:mclachlanlinearsys}
\begin{gathered}
A_{0,0}(t)\frac{\partial \alpha}{\partial t} + \sum_{l=1}^{wn} A_{0,l}(t)\frac{\partial \theta_l}{\partial t} = b_0(t) \\
A_{1,0}(t)\frac{\partial \alpha}{\partial t} + \sum_{l=1}^{wn} A_{1,l}(t)\frac{\partial \theta_l}{\partial t} = b_1(t) \\
\vdots \\
A_{wn,0}(t)\frac{\partial \alpha}{\partial t} + \sum_{l=1}^{wn} A_{wn,l}(t)\frac{\partial \theta_l}{\partial t} = b_{wn}(t), \\
\end{gathered}
\end{equation}
where the parameter derivatives ${\partial \alpha}/{\partial t}, {\partial \theta_1}/{\partial t},  \dots {\partial \theta_{wn}}/{\partial t}$ can be obtained as the solutions. The real-valued coefficients 
\begin{equation}
    A_{k,l}(t) = \text{Re}\left\{\frac{\partial \langle\Tilde{v}(\boldsymbol{\Theta}(t))|}{\partial \theta_l} \frac{\partial |\Tilde{v}(\boldsymbol{\Theta}(t))\rangle}{\partial \theta_k} \right\}
\end{equation}
and
\begin{equation}
    b_{k}(t) = \text{Re}\left\{ \frac{\partial \langle\Tilde{v}(\boldsymbol{\Theta}(t))|}{\partial \theta_l} L_h |\Tilde{v}(\boldsymbol{\Theta}(t))\rangle \right\}
\end{equation}
are computed prior to solving the linear system in Eq. (\ref{eq:mclachlanlinearsys}). In total, $(wn+1)^2 + 4wn$ circuits are required to compute all coefficients. The entries of the coefficient matrix $A$ are computed with $(wn+1)^2$ circuits, whereas the entries of the coefficient vector $\boldsymbol{b}$ are computed with $4wn$ circuits. The circuits to compute $\boldsymbol{b}$ are implemented with a parallel Pauli operation strategy, where the circuits for coefficients are constructed with ancillary qubits which allow $L_h$ to be expanded as a linear combination of Pauli strings. This is a significant improvement over the original VQS \cite{yuan2019theory}, where each $b_k$ is computed with $4^d$ circuits as the worst-case scenario.  After all entries of $A$ and $\boldsymbol{b}$ are obtained, the linear system is solved for ${\partial \alpha}/{\partial t}$ and ${\partial \theta_l}/{\partial t}$'s on a classical computer.

The time derivatives are then used to evolve $\alpha$ and $\boldsymbol{\theta}$ over a time step $\Delta t$. The parameter update rules are based on the Euler's forward method, as
\begin{equation}
    \boldsymbol{\theta}(t + \Delta t) = \boldsymbol{\theta}(t) + \Delta t \frac{\partial \boldsymbol{\theta}}{\partial t},
\end{equation}
and
\begin{equation}
    \alpha(t + \Delta t) = \alpha(t) + \Delta t \frac{\partial \alpha}{\partial t}.
\end{equation}
The parameterized circuit in Eq. (\ref{eq:vqsansatz}) is executed with the updated $\boldsymbol{\theta}(t + \Delta t)$ to obtain $|\Tilde{v}^{(j)}(t + \Delta t)\rangle$.  The amplitudes of  $|\Tilde{v}^{(j)}(t + \Delta t)\rangle$ are multiplied by the updated $\alpha(t + \Delta t)$ to obtain the homogeneous solution $\Tilde{v}^{(j)}(t + \Delta t)$.

\subsection{VQS Circuit Depth Selection}

The searching behavior during the Hilbert space exploration depends on the VQS circuit depth $w$.  As $w$ increases, the number of parameters in the parameterized circuits increases, which also implies that a larger number of circuits is required to solve the linear system in Eq. (\ref{eq:mclachlanlinearsys}).  Therefore, it is critical to select $w$ so that the Hilbert space is sufficiently explored with the minimal computational overhead.

The extent of Hilbert space exploration can be quantified by the Fubini-Study metric \cite{stokes2020quantum}. The Fubini-Study metric of the circuit in Eq. (\ref{eq:vqsansatz}) is
\begin{equation} \label{eq:fstensor}
        {G}(\boldsymbol{\theta}) =
        \begin{bmatrix}
        G^{(1)}(\boldsymbol{\theta}) & 0_{n \times n} & \cdots & 0_{n \times n} \\
       0_{n \times n} & G^{(2)}(\boldsymbol{\theta}) & \cdots & 0_{n \times n} \\
        \vdots & \vdots & \ddots & \vdots \\
        0_{n \times n} & 0_{n \times n} & \cdots & G^{(n)}(\boldsymbol{\theta})  \\
        \end{bmatrix},
\end{equation}
where 
\begin{equation}
    \label{eq:covariancematrixsinglelayer}
        {G}^{(s)}(\boldsymbol{\theta}) =
        \begin{bmatrix}
        g_{(s-1)n+1,(s-1)n+1}(\boldsymbol{\theta}) & g_{(s-1)n+1,(s-1)n+2}(\boldsymbol{\theta}) & \cdots & g_{(s-1)n+1,sn}(\boldsymbol{\theta}) \\
       g_{(s-1)n+2,(s-1)n+1}(\boldsymbol{\theta}) & g_{(s-1)n+2,(s-1)n+2}(\boldsymbol{\theta}) & \cdots & g_{(s-1)n+2,sn}(\boldsymbol{\theta}) \\
        \vdots & \vdots & \ddots & \vdots \\
        g_{sn,(s-1)n+1}(\boldsymbol{\theta}) & g_{sn,(s-1)n+2}(\boldsymbol{\theta}) & \cdots & g_{sn,sn}(\boldsymbol{\theta}) \\
        \end{bmatrix}
\end{equation}
is the Fubini-Study metric of the $s$\textsuperscript{th} layer of the circuit. In Eq. (\ref{eq:covariancematrixsinglelayer}), $g_{k,l}(\boldsymbol{\theta})$ is the covariance of the sensitivities of the quantum state with respect to parameters $\theta_k$ and $\theta_l$, which is defined as 
\begin{equation} 
    \label{eq:fsmetric}
    g_{k,l}(\boldsymbol{\theta}) = \frac{\partial\langle\psi(\boldsymbol{\theta})|}{\partial\theta_k}\frac{\partial|\psi(\boldsymbol{\theta})\rangle}{\partial \theta_l} - \langle\psi(\boldsymbol{\theta})|\frac{\partial|\psi(\boldsymbol{\theta})\rangle}{\partial \theta_k}\frac{\partial \langle\psi(\boldsymbol{\theta})|}{\partial \theta_l}|\psi(\boldsymbol{\theta})\rangle.
\end{equation}
The determinant of $G$, denoted as $\det(G)$, is the squared volume density of the Hilbert space explored by the VQS circuit. The volume density quantifies the scale of the Hilbert space relative to the space of circuit parameters.  That is, if the Hilbert space is $g^*$ times larger than the parameter space, then $\det(G) = g^*$. The minimum value of $\det(G)$ is $0$, which occurs when the gradients of state $|\psi(\boldsymbol{\theta})\rangle$ with respect to at least two parameters are equivalent. In this case, the directions of correlated parameters are perfectly aligned with each other. Since $G^{(s)}$ is a sub-matrix of the block diagonal matrix $G$, $\det(G)$ is obtained as
\begin{equation} \label{eq:detg}
    \det(G) = \prod_{s=1}^w {\det{({G}^{(s)})}}.
\end{equation}
\noindent For the VQS circuit in Eq. (\ref{eq:vqsansatz}), $\det(G)$ decreases exponentially as $w$ and $n$ increase. That is, the exploration of the solution becomes exponentially more difficult as either the circuit depth or the number of qubits increases, which is shown in Theorem \ref{thm:expdecay}.
\begin{theorem}
\label{thm:expdecay}
Let $G^{(s)}$ denote the Fubini-Study metric of the $s$\textsuperscript{th} circuit layer. Then it holds that
\begin{equation} \label{eq:detgbound}
    \det\left({G}\right) < \left( \frac{1}{2} \right)^{wn}.
\end{equation}
\begin{proof}
The Fubini-Study metric $G^{(s)}$ is decomposed as $G^{(s)} = D^{(s)} + N^{(s)}$, where $D^{(s)}$ and $N^{(s)}$ are the diagonal and off-diagonal components of $G^{(s)}$, respectively. Equivalently, $\Tilde{G}^{(s)} = I_n + \Tilde{N}^{(s)}$, where $\Tilde{G}^{(s)} = \beta {G}^{(s)}$ and $\Tilde{N}^{(s)} = \beta{N}^{(s)}$ are matrices scaled from ${G}^{(s)}$ and ${N}^{(s)}$. The scaled diagonal matrix $\beta{D}^{(s)}$ is equivalent to identity matrix $I_n$ since all diagonal entries in ${D}^{(s)}$ are identical. This is because the same rotation gates are applied to all qubits. By Jacobi's formula, where  $\ln{(\det(\Tilde{G}^{(s)}))} = \text{Tr}{(\ln{(\Tilde{G}^{(s)})})}$,
the Taylor series expansion of $\ln{(\det(\Tilde{G}^{(s)}))}$ is
\begin{equation} \label{eq:taylorexpansion}
   \ln{(\det(\Tilde{G}^{(s)}))} = \sum_{s = 1}^\infty (-1)^{s+1}\frac{\text{Tr}{((\Tilde{N}^{(s)})^s)}}{s}.
\end{equation}
It is observed that Eq. (\ref{eq:taylorexpansion}) is dominated by the second-order term $-{\text{Tr}{((\Tilde{N}^{(s)})^2)}}/2$ since $\text{Tr}{(\Tilde{N}^{(s)})} = 0$. It is also observed that $\text{Tr}{((\Tilde{N}^{(s)})^2)} \geq 0$ since $\Tilde{N}$ is a symmetric matrix. Therefore, $\ln(\det(\Tilde{G}^{(s)})) \leq 0$, which is equivalent to $\det(\Tilde{G}^{(s)}) \leq 1$. Subsequently, $\det({G}^{(s)}) \leq \beta^{-n}$. Given the circuit in Eq. (\ref{eq:vqsansatz}), the derivative of $\left|\psi\left(\boldsymbol{\theta}\right)\right\rangle$ with respect to $\theta_{k}$ is
\begin{equation} \label{eq:statederivative}
    \frac{\partial\left|\psi\left(\boldsymbol{\theta}\right)\right\rangle}{\partial \theta_k} = -\frac{i}{2}\left[ \prod_{s=\lceil k/n \rceil+1}^w  \left( U_e  U_r(\boldsymbol{\theta}^{(\lceil k/n \rceil+1+w-s)}) \right) \right] U_e  P_kU_r(\boldsymbol{\theta}^{(\lceil k/n \rceil}) \left[ \prod_{s=1}^{\lceil k/n \rceil-1}  \left( U_e  U_r(\boldsymbol{\theta}^{(\lceil k/n \rceil-s)}) \right) \right] |0\rangle^{\otimes n},
\end{equation}
where $P_k$ is a $2^n \times 2^n$ matrix defined as
\begin{equation} \label{eq:pauliblockmatrix}
    P_k = I_{2^{k (\text{mod} n) - 1}} \otimes \sigma_y \otimes I_{2^{n - k (\text{mod} n)}}.
\end{equation}
By substitution of Eq. (\ref{eq:statederivative}) into Eq. (\ref{eq:fsmetric}), each entry in $G^{(s)}$ is
\begin{equation} \label{eq:fsentry}
    g_{k,l} = \frac{1}{4} (\left\langle  P_kP_l \right\rangle - \left\langle P_k \right\rangle \left\langle P_l  \right\rangle),
\end{equation}
where $\left\langle P_kP_l \right\rangle$, $\left\langle  P_k \right\rangle$, and $\left\langle P_l \right\rangle$ are expectation values with respect to $\left|\psi\left(\boldsymbol{\theta}\right)\right\rangle$. It is observed that the maximum and minimum values of $\langle P_k \rangle$, $\langle P_l \rangle$, and $\langle P_kP_l \rangle$  are 1 and $-1$. It follows that the maximum value of any entry in ${G}^{(s)}$ is $1/2$. Therefore, $\beta^{-n} \leq (1/2)^n$. Subsequently,
\begin{equation} \label{eq:detgsbound}
    \det({G}^{(s)}) \leq \left(\frac{1}{2}\right)^n.
\end{equation}
By substitution of Eq. (\ref{eq:detgsbound}) into Eq. ({\ref{eq:detg}}), Eq. (\ref{eq:detgbound}) is obtained.
\end{proof}
\end{theorem}

\section{Simulation Examples} \label{sec:results}
\subsection{Example \#1: Vorticity Transport Equation}
The vorticity transport equation is a second-order nonlinear PDE which describes the time evolution of local fluidic rotations. The PDE is formulated as
\begin{equation} \label{eq:vorticitytransport}
\frac{\partial \omega}{\partial t}  = \frac{\partial \phi}{\partial x} \frac{\partial \omega}{\partial y} -  \frac{\partial \phi}{\partial y} \frac{\partial \omega}{\partial x} + \nu \left( \frac{\partial^2 \omega}{\partial x^2} + \frac{\partial^2 \omega}{\partial y^2} \right),
\end{equation}
where $\nu$ is the kinematic viscosity, $t$ is time, and $x$ and $y$ are Cartesian coordinates. The solutions to Eq. (\ref{eq:vorticitytransport}) include the vorticity $\omega$ and streamfunction $\phi$. Both $\omega$ and $\phi$ are related by Poisson's equation, which is
\begin{equation} \label{eq:poisson}
\frac{\partial^2 \phi}{\partial x^2} + \frac{\partial^2 \phi}{\partial y^2} + \omega = 0.
\end{equation}

In this example, $\omega$ and $\phi$ are computed at multiple points arranged into a square grid with a side length of 1 m. The values of $\omega$ and $\phi$ in an $8 \times 8$ grid are obtained with VQHPM.  The initial vorticity field is $\omega(x, y, t=0) = 0.75\exp(-40((x-0.5)^2 + (y-0.5)^2))$. Given $\omega(x, y, t=0)$, the initial streamfunction field $\phi(x, y, t=0)$ is obtained by solving Eq. (\ref{eq:poisson}) with FDM.  The vorticity transport equation is solved for 5 time steps, where each time step occurs for $\Delta t = 0.0002$ s. At each time step, the $m+1$ homotopy series terms ${\omega}_v^{(j)}$'s in ${\omega} \approx {\omega}_v^{(0)} + \sum_{j=1}^m {\omega}_v^{(j)}$ and $m+1$ homotopy series terms ${\phi}_v^{(j)}$'s in ${\phi} \approx {\phi}_v^{(0)} + \sum_{j=1}^m {\phi}_v^{(j)}$ are obtained by solving the initial linear deformation equation \begin{equation} \label{eq:zerothvorticitydeformationeq}
    \frac{\partial {\omega}_v^{(0)}}{\partial t} = \nu \left( \frac{\partial^2 {\omega}_v^{(0)}}{\partial x^2} + \frac{\partial^2 {\omega}_v^{(0)}}{\partial y^2} \right)
\end{equation}
and $m$ linear deformation equations of the form
\begin{equation} \label{eq:vorticitydeformationeq}
    \frac{\partial {\omega}_v^{(j)}}{\partial t} = \nu \left( \frac{\partial^2 {\omega}_v^{(j)}}{\partial x^2} + \frac{\partial^2 {\omega}_v^{(j)}}{\partial y^2} \right) + \sum_{k = 1}^{j+1}\sum_{l=1}^{k} \left(  \frac{\partial \phi_v^{(l-1)}}{\partial x} \frac{{\omega}_v^{(k-l)}}{\partial y} - \frac{\partial \phi_v^{(l-1)}}{\partial y} \frac{\partial {\omega}_v^{(k-l)}}{\partial x} \right).
\end{equation}
The initial linear term ${\omega}_v^{(0)}(x,y,t=\Delta t)$ is obtained by solving Eq. (\ref{eq:zerothvorticitydeformationeq}) with VQS, where the initial guess ${\omega}_v^{(0)}(x,y,t=0)$ is set to the nonlinear initial condition ${\omega}(x,y,t=0)$. After ${\omega}_v^{(0)}(x,y,t=\Delta t)$ is computed, ${\phi}_v^{(0)}(x,y,t=\Delta t)$ is obtained by solving Eq. (\ref{eq:poisson}) with FDM. The nonlinear correction terms ${\omega}_v^{(j)}(x,y,t=\Delta t)$'s and ${\phi}_v^{(j)}(x,y,t=\Delta t)$'s for all $j \geq 1$ are then computed by solving Eq. (\ref{eq:vorticitydeformationeq}), which is decomposed into two sub-problems. One sub-problem is to solve the steady-state PDE
\begin{equation} \label{eq:vorticitysteadystateeq}
    \nu \left( \frac{\partial^2 \hat{{\omega}}_v^{(j)}}{\partial x^2} + \frac{\partial^2 \hat{{\omega}}_v^{(j)}}{\partial y^2} \right) =  -\sum_{k = 1}^{j+1}\sum_{l=1}^{k} \left( \frac{\partial \hat{\phi}_v^{(l-1)}}{\partial y} \frac{\partial \hat{\phi}_v^{(k-l)}}{\partial x} - \frac{\partial \hat{\phi}_v^{(l-1)}}{\partial x} \frac{\hat{{\omega}}_v^{(k-l)}}{\partial y} \right),
\end{equation}
where the solution $\hat{\omega}_v^{(j)}(x,y)$ is the steady-state component of ${\omega}_v^{(j)}(x,y,t=\Delta t)$. The sum on the right side of Eq. (\ref{eq:vorticitysteadystateeq}) depends on steady-state solutions previously obtained from lower-order linear deformation equations. Eq. (\ref{eq:poisson}) is then solved to obtain $\hat{\phi}_v^{(j)}(x,y)$.
The other sub-problem is to solve the homogeneous PDE
\begin{equation} \label{eq:vorticityhomogeneouseq}
    \frac{\partial \Tilde{{\omega}}_v^{(j)}}{\partial t} = \nu \left( \frac{\partial^2 \Tilde{{\omega}}_v^{(j)}}{\partial x^2} + \frac{\partial^2 \Tilde{{\omega}}_v^{(j)}}{\partial y^2} \right)
\end{equation}
with VQS, where the initial condition  $\Tilde{\omega}_v^{(j)}(x,y,t=0)$ is set to $-\hat{\omega}_v^{(j)}(x,y)$ since ${\omega}_v^{(j)}(x,y,t) = 0$ for all $j \geq 1$. The solution $\Tilde{\phi}_v^{(j)}(x,y,t=0)$ is then obtained by solving Eq. (\ref{eq:poisson}). After all homotopy series terms are computed, the nonlinear solutions $\omega$ and $\phi$ in the $8 \times 8$ grid are approximated as $\omega(x,y,t=\Delta t) \approx {\omega}_v^{(0)}(x,y,\Delta t) + \sum_{j=1}^m \left[\hat{\omega}_v^{(j)}(x,y) + \Tilde{\omega}_v^{(j)}(x,y,t=\Delta t) \right]$ and $\phi(x,y,t=\Delta t) \approx {\phi}_v^{(0)}(x,y,\Delta t) + \sum_{j=1}^m \left[\hat{\phi}_v^{(j)}(x,y) + \Tilde{\phi}_v^{(j)}(x,y,\Delta t) \right]$. The values of $\omega$ and $\phi$ on the 8 $\times$ 8 grid are then used to approximate solution values in a 100 $\times$ 100 grid with the Chebyshev spectral collocation method. At coordinates $x$ and $y$, $\omega(x,y,t) \approx \sum_{k=1}^{8} \sum_{l=1}^{8}  \bar{\omega}(x_k, y_l, t)\xi(x_k)\xi(y_l)$ and $\phi(x,y,t) \approx \sum_{k=1}^{8} \sum_{l=1}^{8}  \bar{\phi}(x_k, y_l, t)\xi(x_k)\xi(y_l)$ are approximated as linear combinations of products of two interpolation functions $\xi(x_k)$ and $\xi(y_l)$, where $\xi$ is a Chebyshev polynomial of the first kind. Each $\xi(x_k)\xi(y_l)$ is scaled with coefficients $\bar{\omega}(x_k, y_k, t)$ and $\bar{\phi(}x_k, y_k, t)$, which are discrete cosine transformations of $\omega(x_k, y_k, t)$ and $\phi(x_k, y_k, t)$, respectively. 

\begin{figure}[h!] 
        \centering
    \begin{subfigure}{0.45\textwidth}
        \centering
        \includegraphics[width=\linewidth, clip]{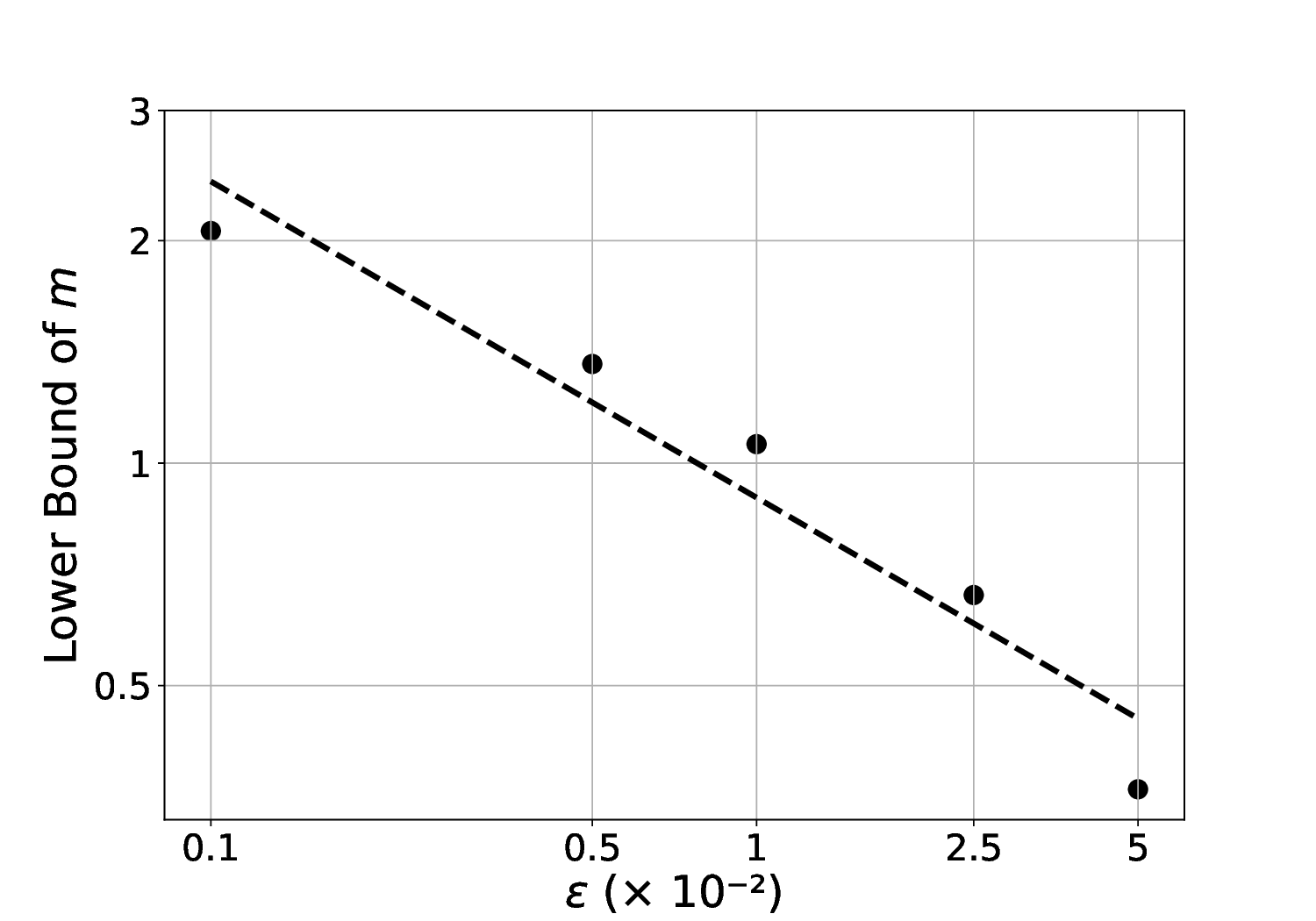}
        \caption{}
        \label{subfig:vorticitytransport_truncationerror_omega}
    \end{subfigure}
    \begin{subfigure}{0.45\textwidth}
        \centering
        \includegraphics[width=\linewidth, clip]{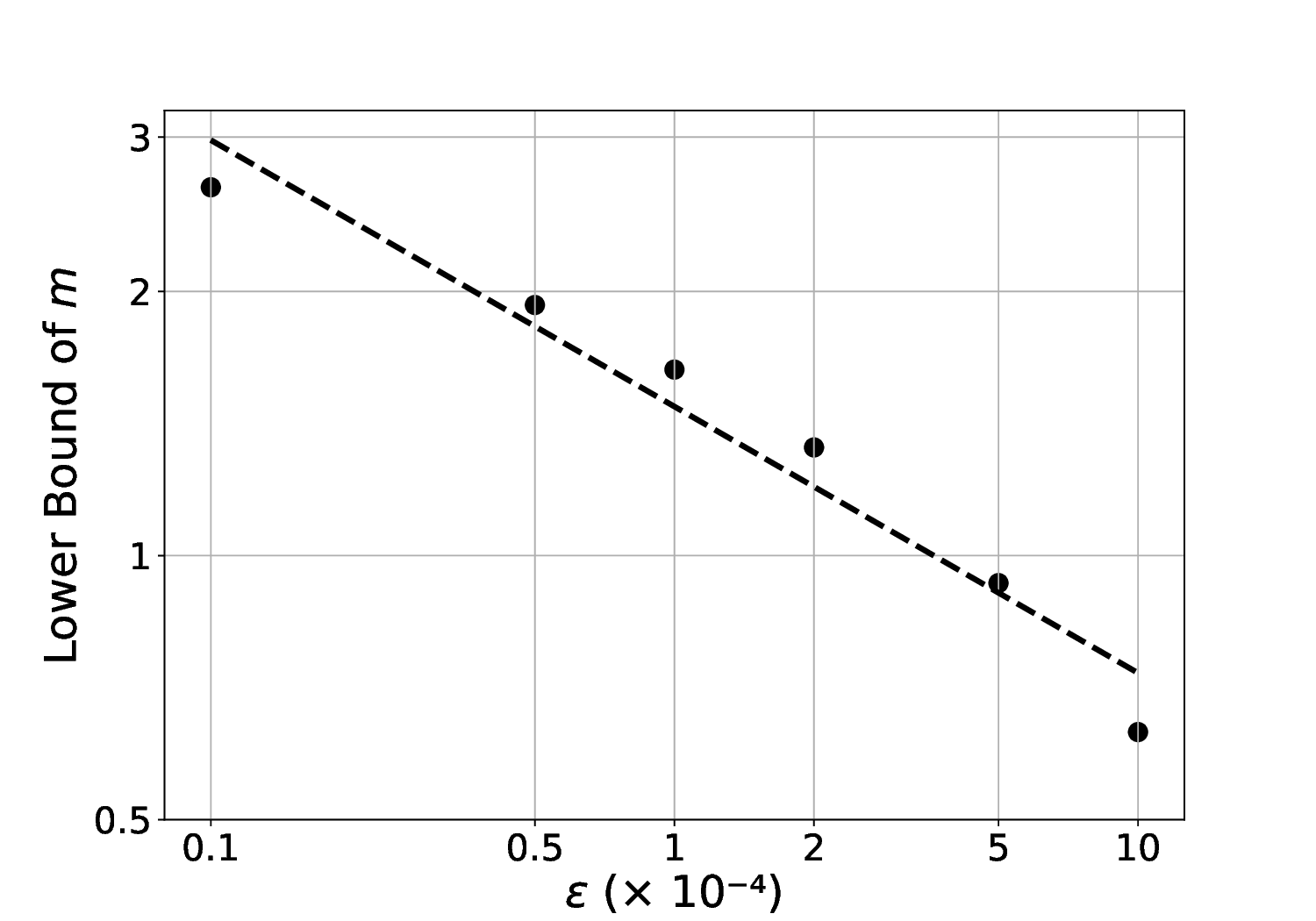}
        \caption{}
        \label{subfig:vorticitytransport_truncationerror_phi}
    \end{subfigure}
    \caption{Lower bound of Corollary \ref{thm:truncation} vs. targeted approximation error of calculating (a) vorticity and (b) streamfunction in the vorticity transport equation}
    \label{fig:vorticitytransport_truncationvserror}
\end{figure}

The targeted approximation errors are set to $\epsilon_\omega = 0.0375$ for $\omega$ and $\epsilon_\phi = 0.0001875$ for $\phi$.  Both targeted errors are 5\% of 0.75 and 0.00375, which are the largest absolute values of the initial $\omega$ and $\phi$ fields, respectively.  The values of $\epsilon_\omega$ and $\epsilon_\phi$ are used to select the  homotopy order $m$ in Figure \ref{fig:vorticitytransport_truncationvserror}, which illustrates a plot of the lower bound in Corollary \ref{thm:truncation}. It is assumed that the contractive ratio $q$ for both $\omega$ and $\phi$ is 0.1, which corresponds to fast convergences of homotopy series solutions towards the nonlinear solutions.  In Figures \ref{subfig:vorticitytransport_truncationerror_omega} and $\ref{subfig:vorticitytransport_truncationerror_phi}$, the lower bounds are obtained as 0.487 and 1.356.  Therefore, $m$ is set to 2.

\begin{figure}[h!] 
    \centering
    \includegraphics[width=0.5\linewidth]{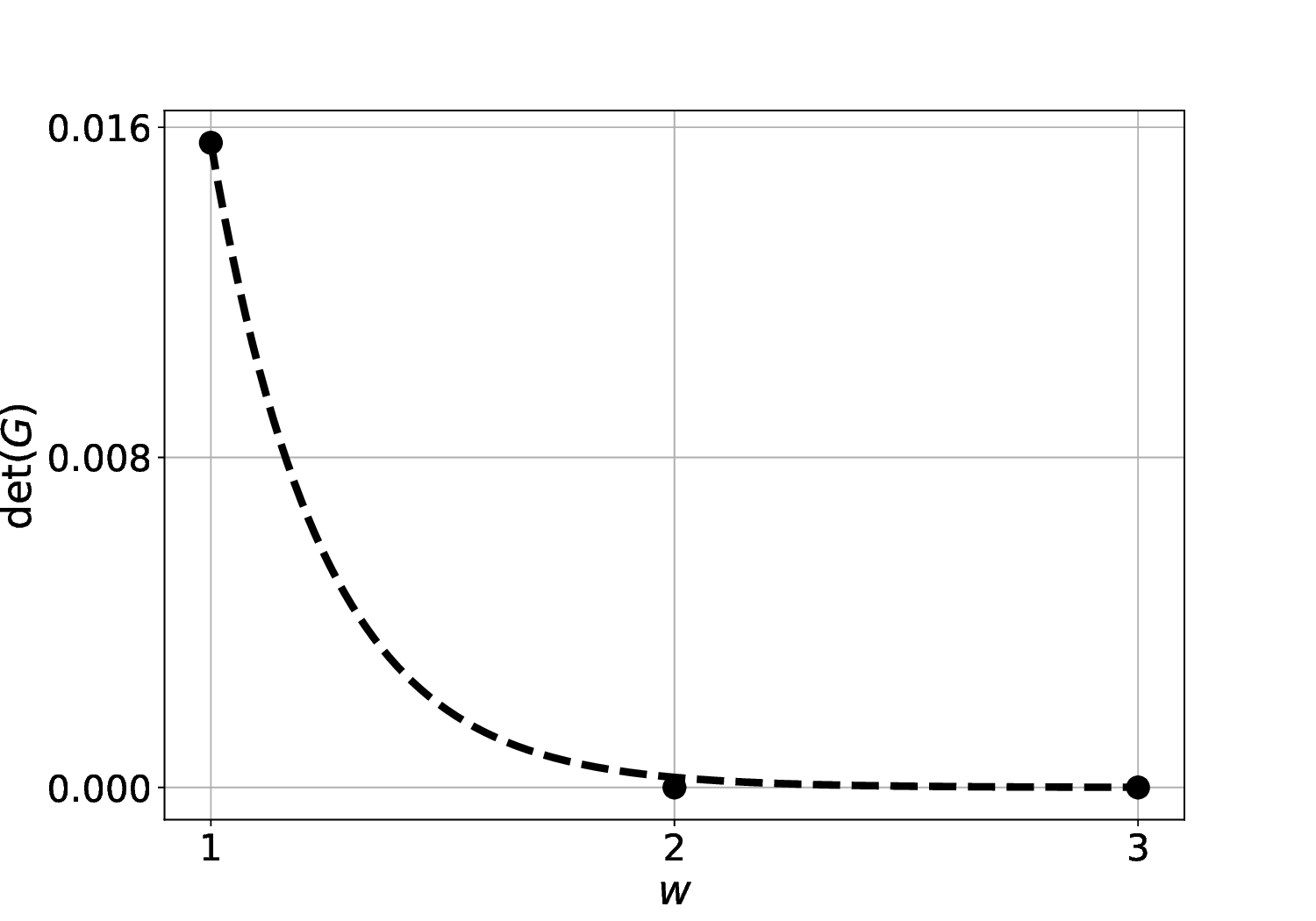}
    \caption{Estimated values of det$(G)$ at different circuit depths for the vorticity transport equation}
    \label{fig:vorticitytransport_fsvsdepth}
\end{figure}

The circuit depth $w$ is selected based on the determinant of the Fubini-Study metric $G$. For each value of $w$ from 1 to 3, $\det(G)$ is estimated as the average value of Eq. (\ref{eq:detg}) from 500 random samples of $\boldsymbol{\theta}$. The estimated values of $\det(G)$ are depicted as black dots in Figure \ref{fig:vorticitytransport_fsvsdepth}. According to the dashed curve, $\det(G)$ exponentially decreases from 0.0156 to 0 as $w$ increases from 1 to 2. To minimize the number of redundant circuit parameters, $w$ is set to 1.

\begin{figure}[h!]
    \centering
    \begin{subfigure}{0.32\textwidth}
        \centering
        \includegraphics[width=\linewidth, trim={1.5cm 1.5cm 1.5cm 1.5cm}, clip]{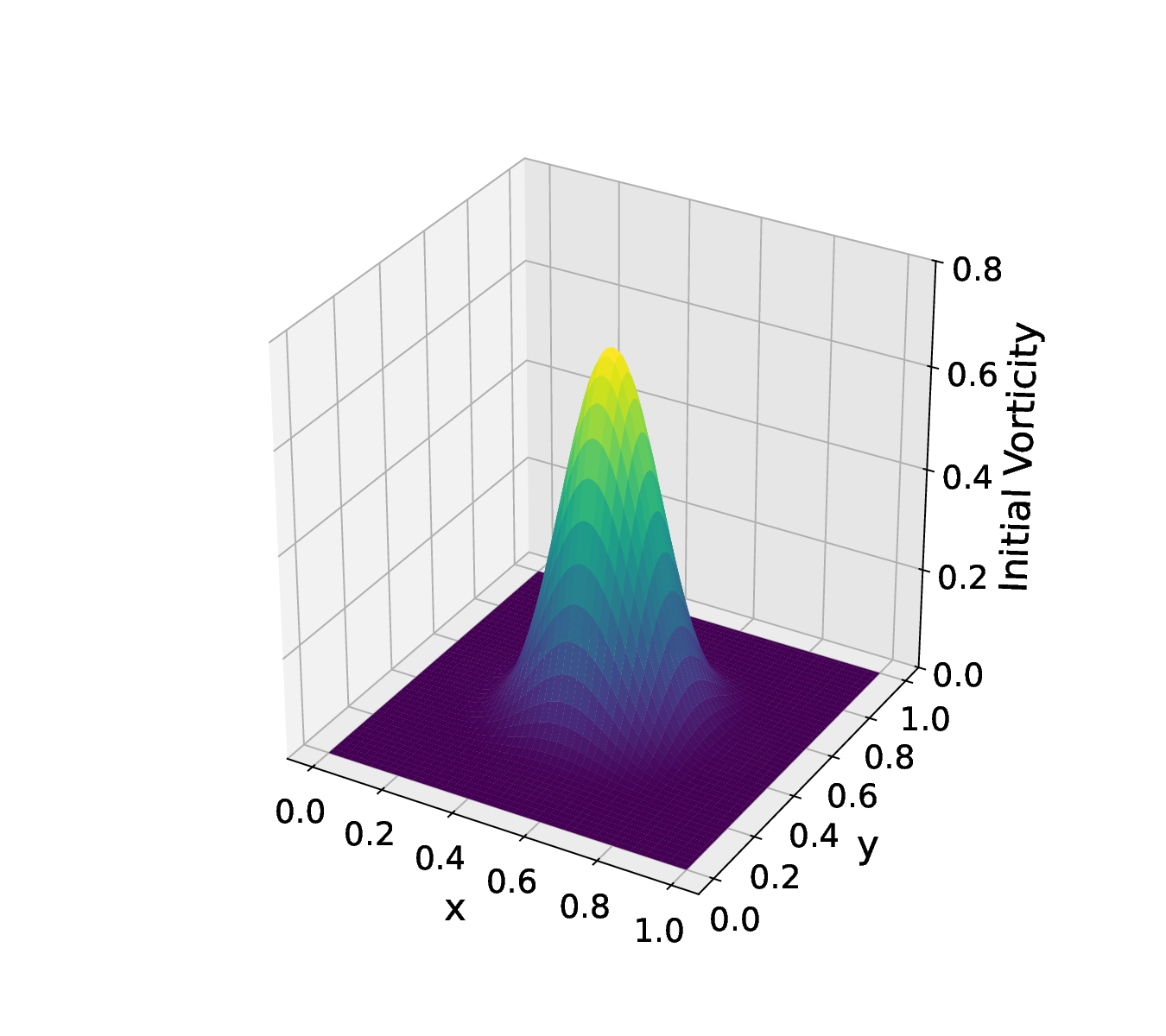}
        \caption{}
        \label{subfig:vorticitytransport_initialvorticity}
    \end{subfigure}
    \begin{subfigure}{0.32\textwidth}
        \centering
        \includegraphics[width=\linewidth, trim={1.5cm 1.5cm 1.5cm 1.5cm}, clip]{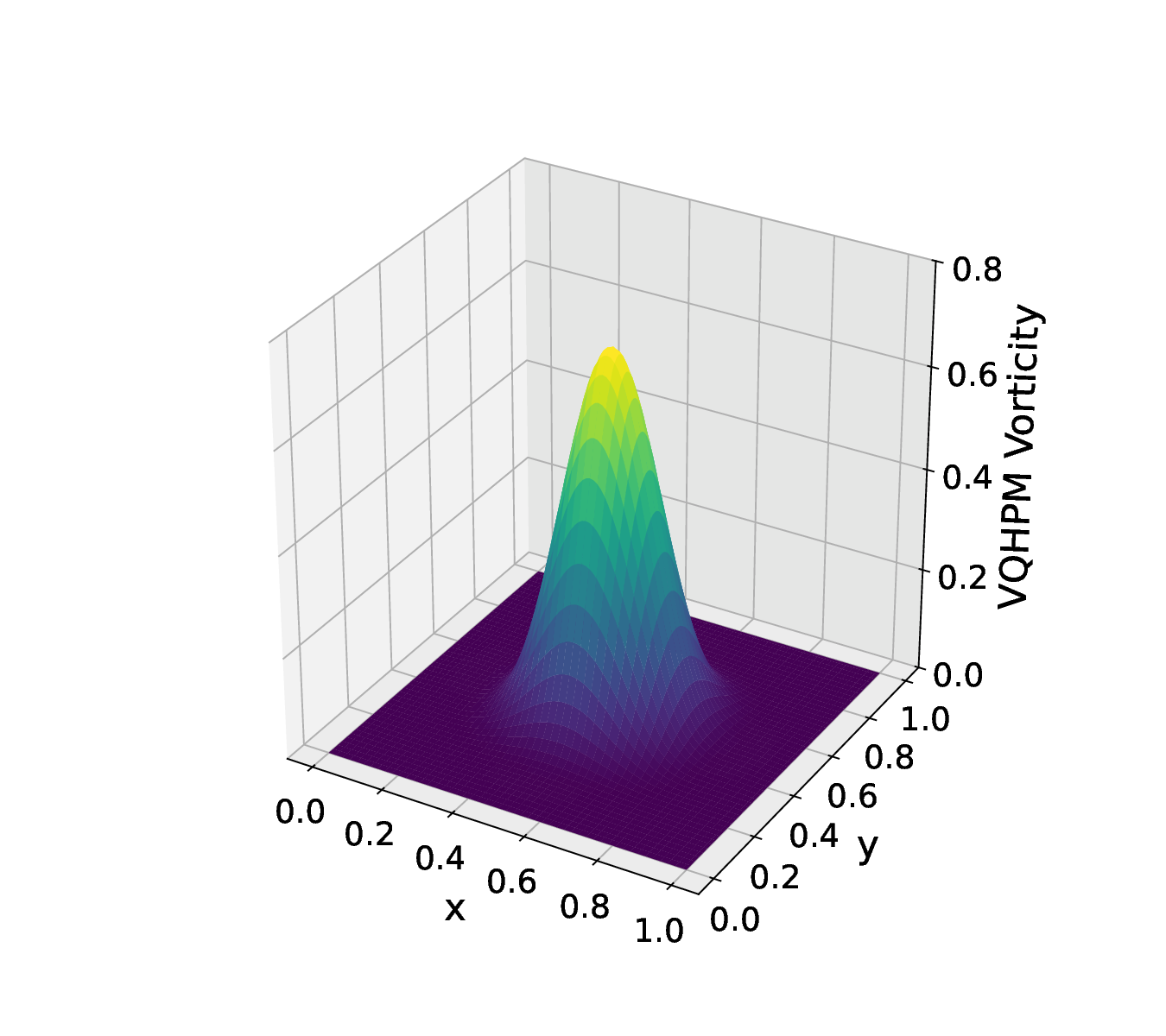}
        \caption{}
        \label{subfig:vorticitytransport_finalvorticity_qhpm}
    \end{subfigure}
    \begin{subfigure}{0.32\textwidth}
        \centering
        \includegraphics[width=\linewidth, trim={1.5cm 1.5cm 1.5cm 1.5cm}, clip]{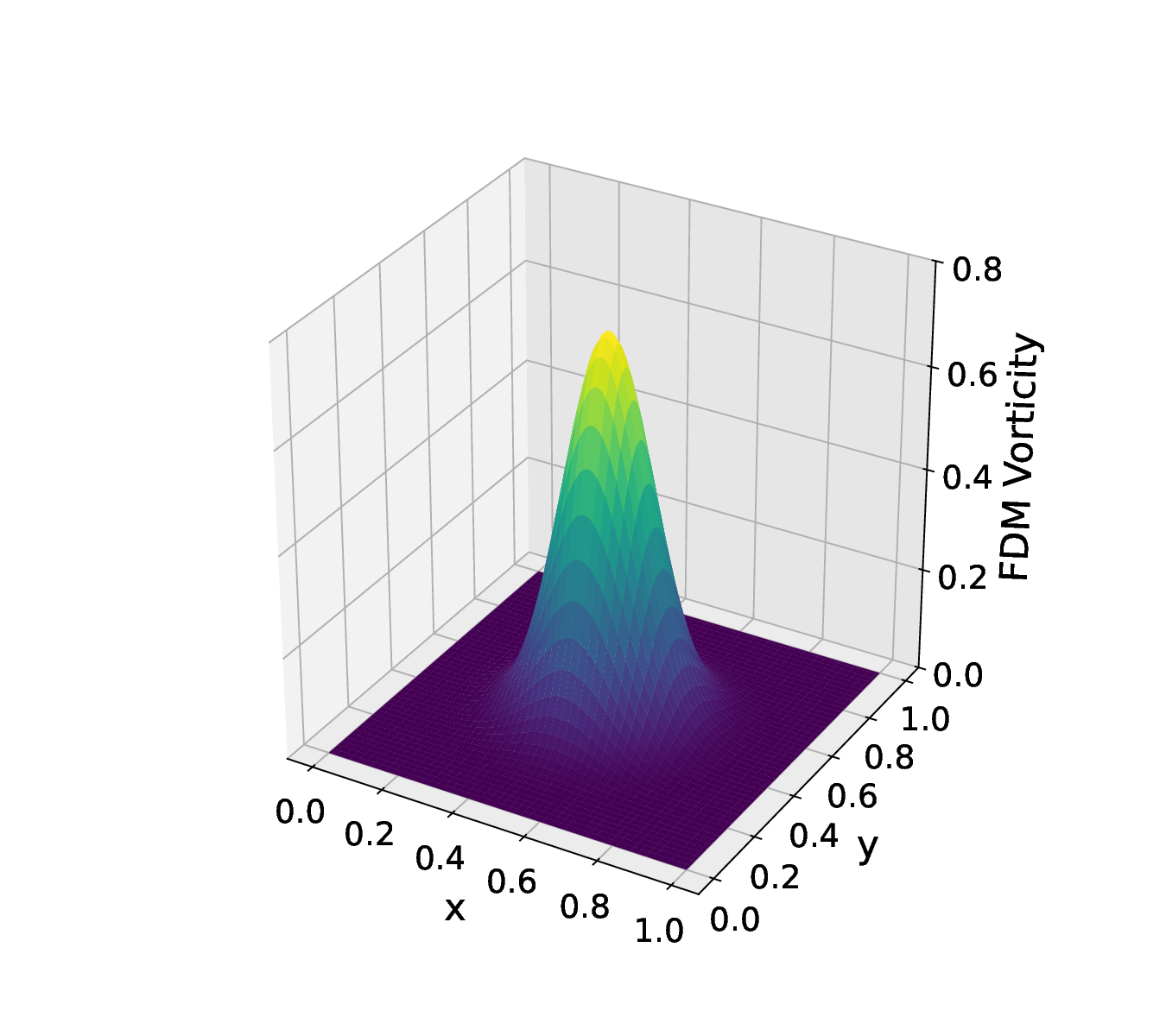}
        \caption{}
        \label{subfig:vorticitytransport_finalvorticityfdm}
    \end{subfigure}
    \begin{subfigure}{0.32\textwidth}
        \centering
        \includegraphics[width=\linewidth, trim={0.75cm 1.5cm 0.75cm 1.5cm}, clip]{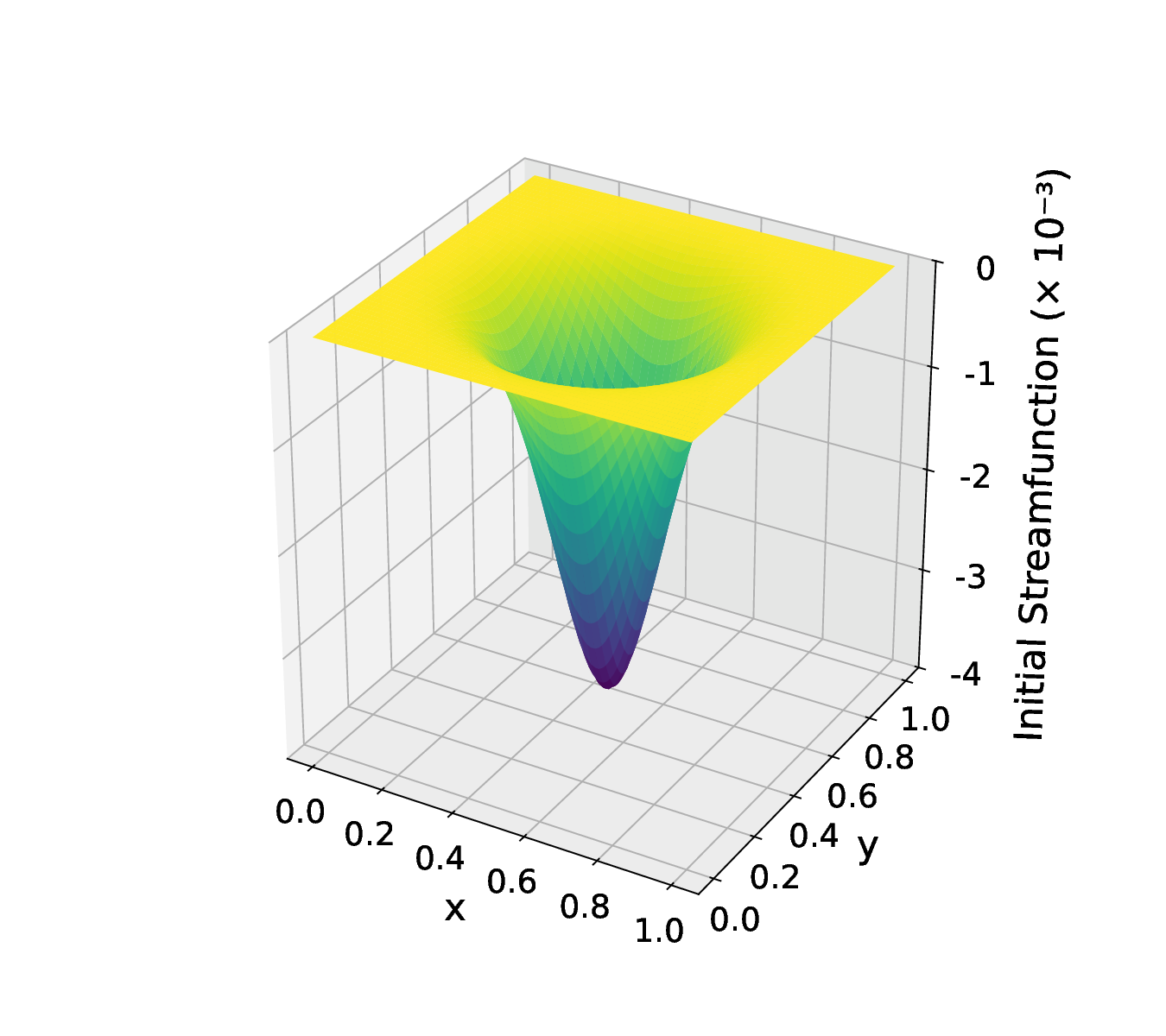}
        \caption{}
        \label{subfig:vorticitytransport_initialstreamfunction}
    \end{subfigure}
    \begin{subfigure}{0.32\textwidth}
        \centering
        \includegraphics[width=\linewidth, trim={0.75cm 1.5cm 0.75cm 1.5cm}, clip]{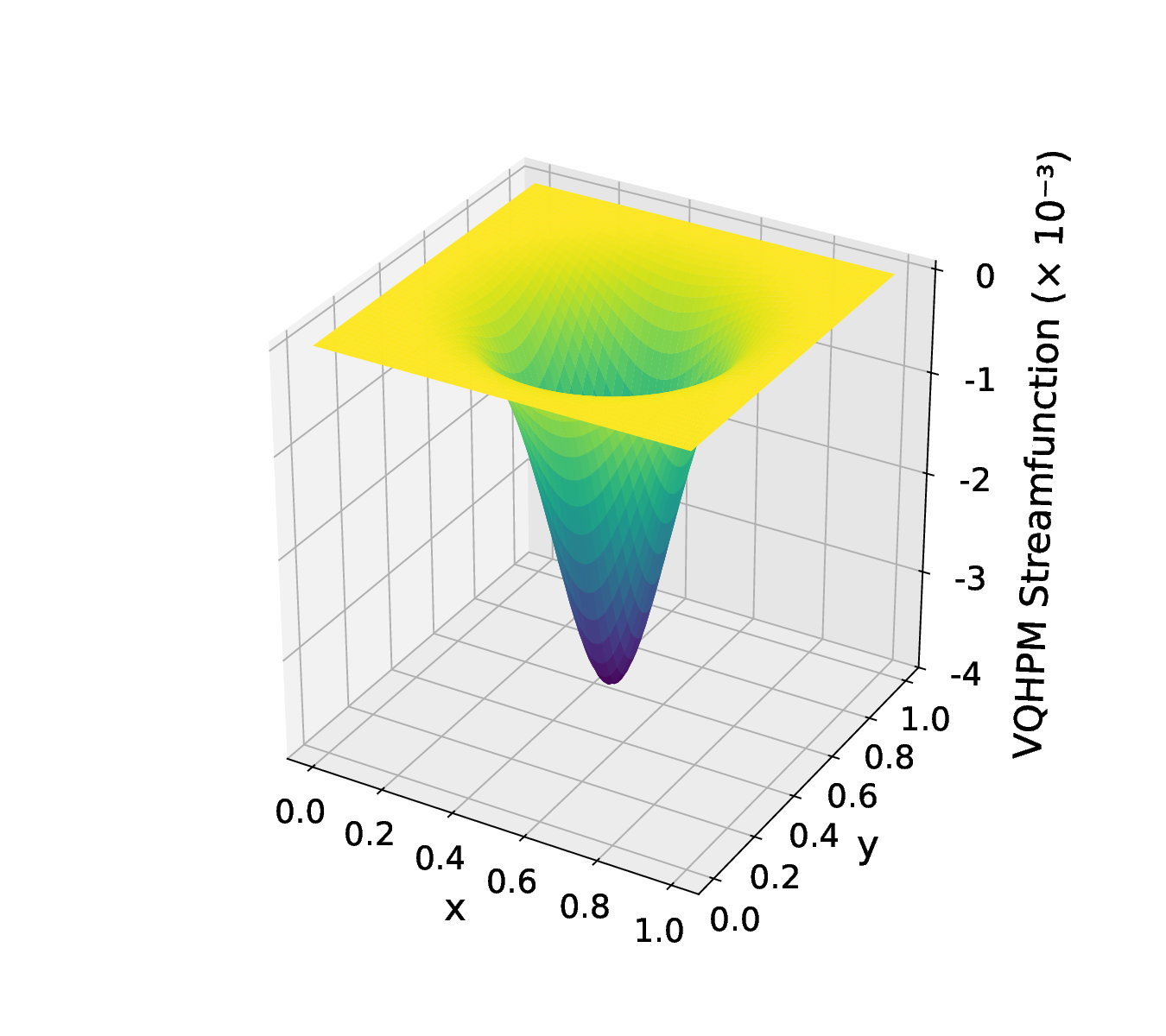}
        \caption{}
        \label{subfig:vorticitytransport_finalstreamfunctionqhpm}
    \end{subfigure}
    \begin{subfigure}{0.32\textwidth}
        \centering
        \includegraphics[width=\linewidth, trim={0.75cm 1.5cm 0.75cm 1.5cm}, clip]{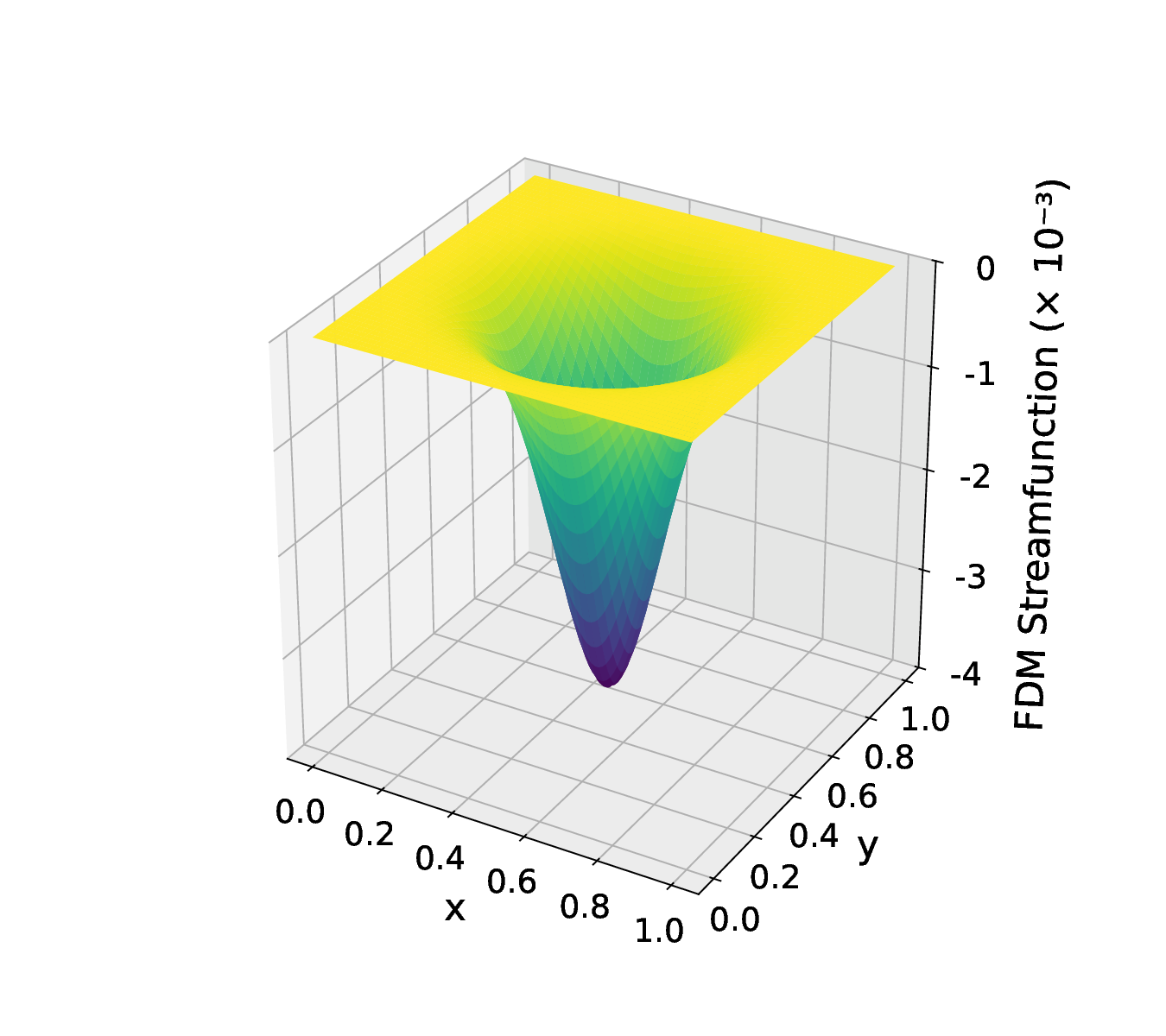}
        \caption{}
        \label{subfig:vorticitytransport_finalstreamfunctionfdm}
    \end{subfigure}
    \caption{Vorticity and streamfunction fields for the vorticity transport equation, including (a, d) initial condition at t = 0, (b, e) VQHPM results at t = 0.001 s, and (c, f) FDM results at t = 0.001 s}
    \label{fig:vorticitytransport_finalfields}
\end{figure}

\begin{figure}[h!] \label{fig:NSvorticityerrorfields}
    \centering
    \begin{subfigure}{0.45\textwidth}
        \centering
        \includegraphics[width=\linewidth, trim={0.25cm 1.5cm 0.25cm 1.5cm}, clip]{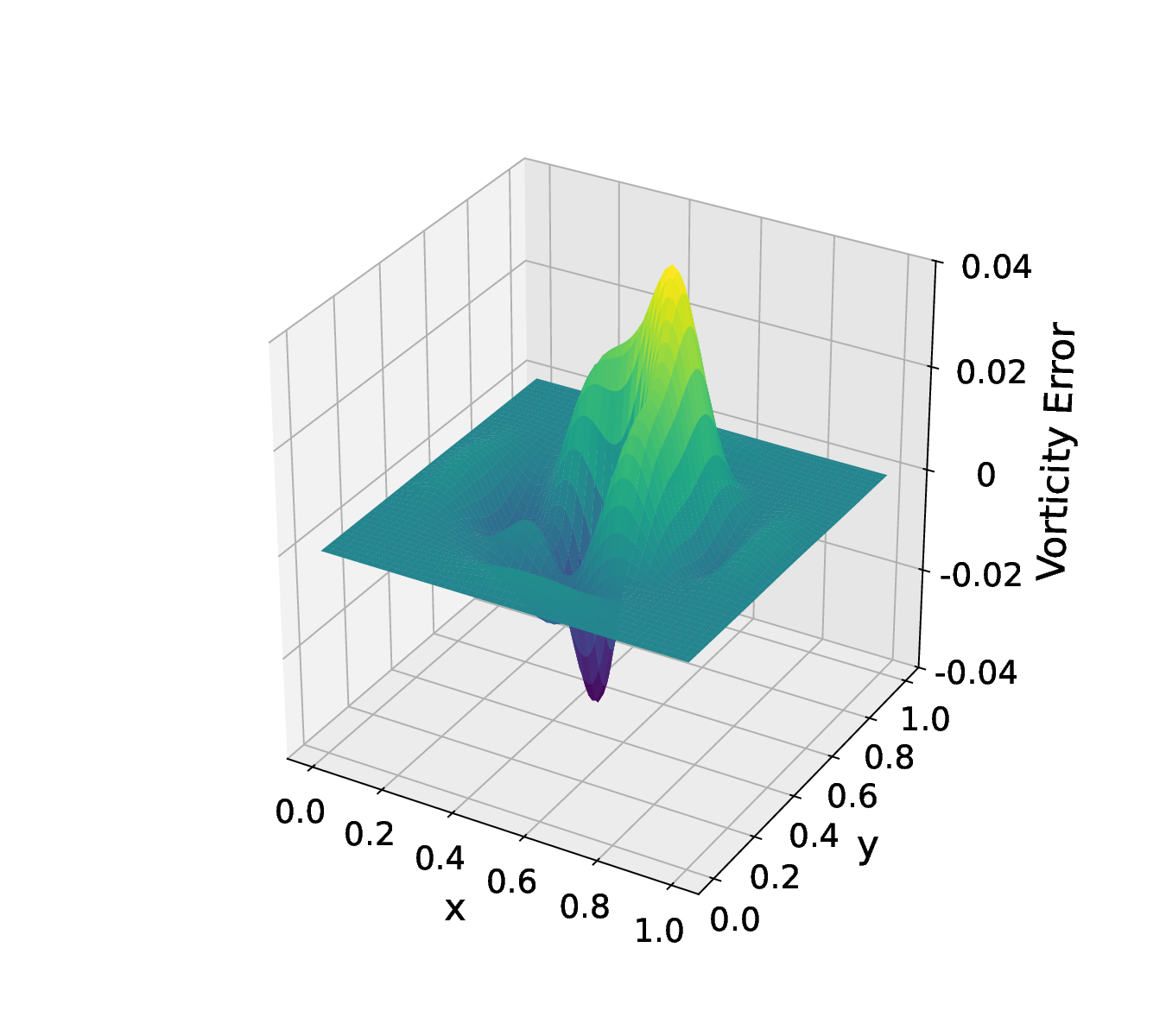}
        \caption{}
        \label{subfig:vorticitytransport_vorticityerror}
    \end{subfigure}
    \begin{subfigure}{0.45\textwidth}
        \centering
        \includegraphics[width=\linewidth, trim={0.25cm 1.5cm 0.25cm 1.5cm}, clip]{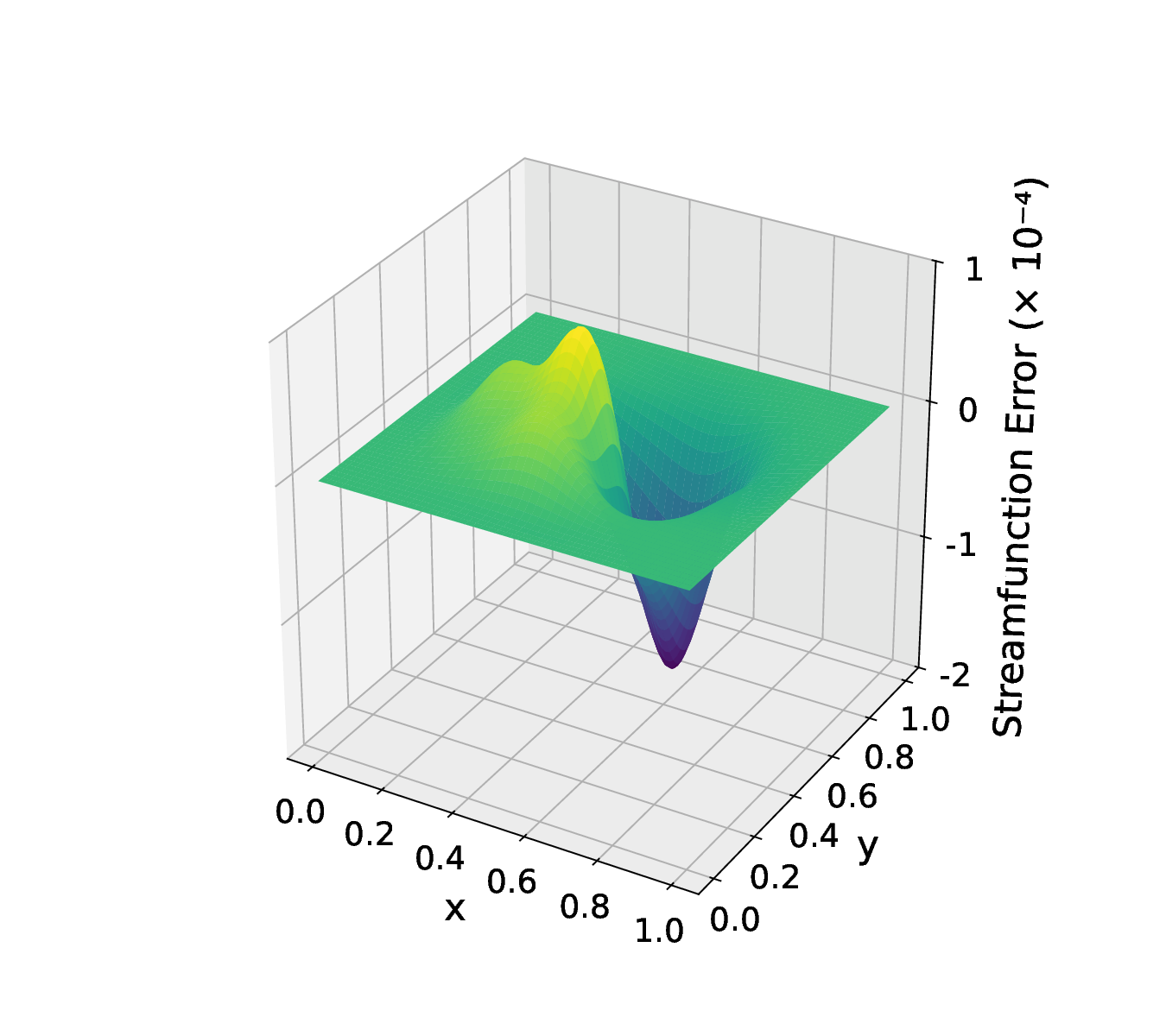}
        \caption{}
        \label{subfig:vorticitytransport_streamfunctionerror}
    \end{subfigure}
    
    \caption{(a) Vorticity and (b) streamfunction errors of solving vorticity transport equation with VQHPM relative to FDM}
    \label{fig:vorticitytransport_errorfields}
\end{figure}

The vorticity transport equation is solved with VQHPM and FDM for the purpose of comparing solution fields, which are illustrated in Figure \ref{fig:vorticitytransport_finalfields}.  Given the initial conditions for $\omega$ and $\phi$ in Figures \ref{subfig:vorticitytransport_initialvorticity} and \ref{subfig:vorticitytransport_initialstreamfunction}, both methods result in highly similar values of $\omega$ at $t = 0.001$ in Figures \ref{subfig:vorticitytransport_finalvorticity_qhpm} and \ref{subfig:vorticitytransport_finalvorticityfdm}, and $\phi$ at $t = 0.001$ in Figures \ref{subfig:vorticitytransport_finalstreamfunctionqhpm} and \ref{subfig:vorticitytransport_finalstreamfunctionfdm}.  According to the vorticity error plot in Figure \ref{subfig:vorticitytransport_vorticityerror}, the absolute differences between the VQHPM and FDM vorticities at most grid points are at most 0.0375, which is the selected value of $\epsilon_\omega$. Only 88 out of 10,000 grid points result in absolute differences which are larger than $\epsilon_\omega$, where the maximum absolute difference is 0.045. In the streamfunction error plot in Figure \ref{subfig:vorticitytransport_streamfunctionerror}, the maximum absolute difference between all grid points is 0.00017, which is less than $\epsilon_\phi$.

The convergence behavior of VQHPM is also analyzed.  Figures \ref{subfig:vorticitytransport_averageabsolutevorticityerror} and \ref{subfig:vorticitytransport_averageabsolutestreamfunctionerror} illustrate the average absolute errors for $\omega$ and $\phi$ at different $m$ and $t$.  The average errors are defined as the average values of absolute differences at all 10,000 grid points.   At each $t$, the average absolute errors for $\omega$ and $\phi$ initially decrease from $m = 1$ to $m = 2$, then remain constant from $m = 2$ to $m = 3$. This means that the choice of $m = 2$ is optimal for maximizing approximation accuracy and minimizing the number of linear deformation equations to solve.

\begin{figure}[h!] 
    \centering
    \begin{subfigure}{0.45\textwidth}
        \centering
        \includegraphics[width=\linewidth]{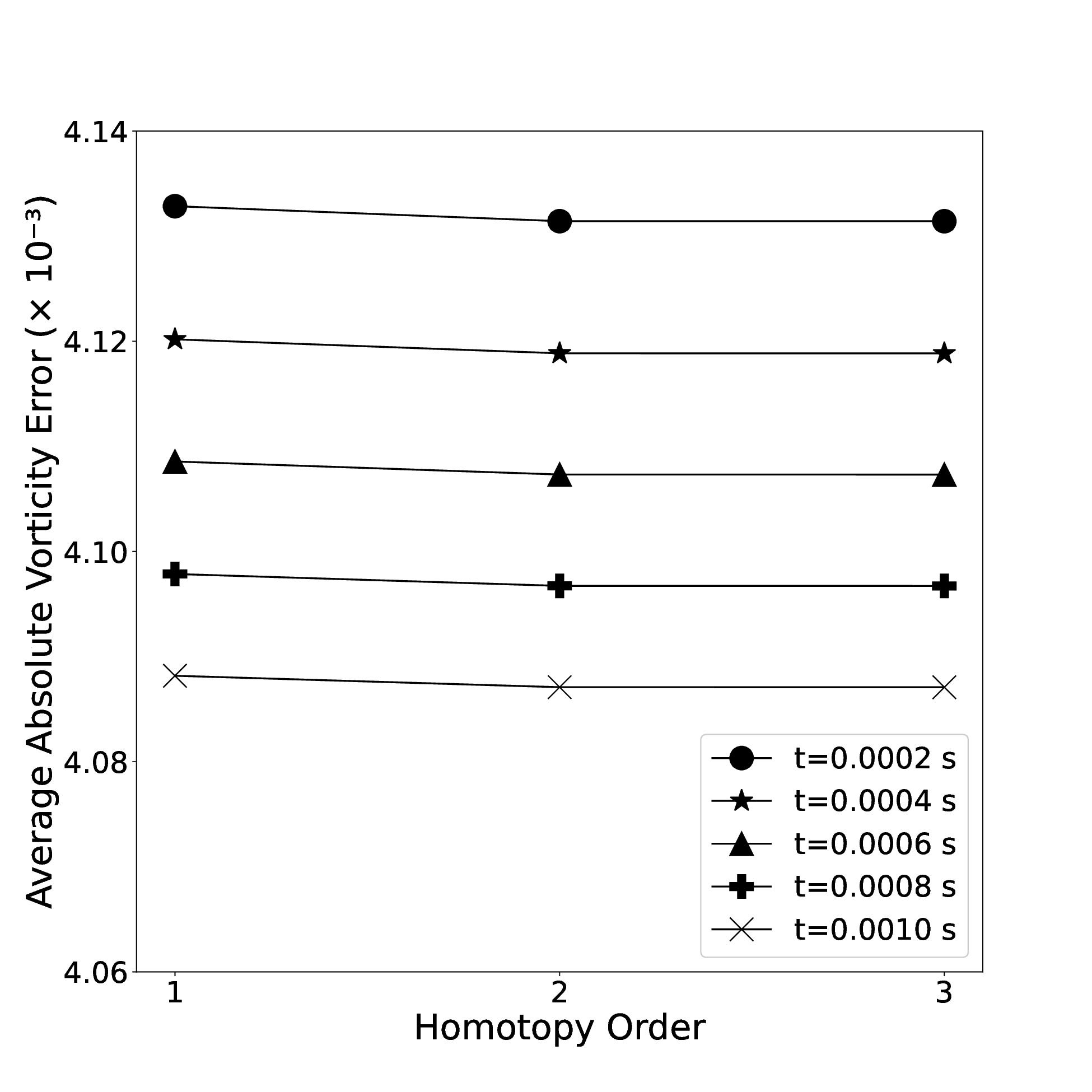}
        \caption{}
        \label{subfig:vorticitytransport_averageabsolutevorticityerror}
    \end{subfigure}
    \begin{subfigure}{0.45\textwidth}
        \centering
        \includegraphics[width=\linewidth]{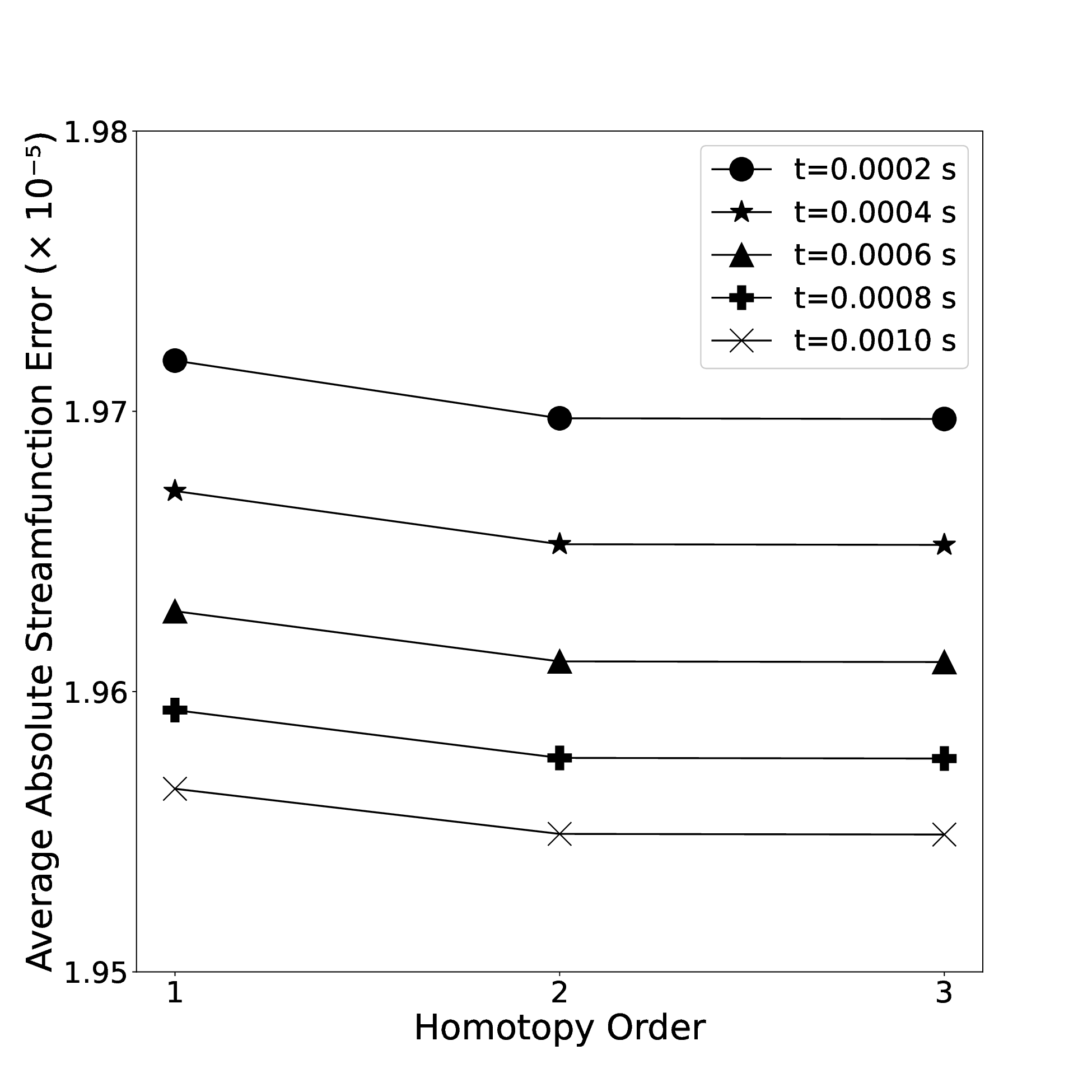}
        \caption{}
        \label{subfig:vorticitytransport_averageabsolutestreamfunctionerror}
    \end{subfigure}
    
    \caption{(a) Average vorticity and (b) streamfunction errors of solving vorticity transport equation with VQHPM relative to FDM for different homotopy orders}
    \label{fig:vorticitytransport_averageabsoluteerrors}
\end{figure}

\subsection{Example \#2: Reduced Magnetohydrodynamics}
The dynamics of an electrically conductive fluid is described with reduced magnetohydrodynamics (MHD). The reduced MHD equations are a system of two nonlinear PDEs and two linear PDEs formulated as
\begin{equation} \label{eq:MHD}
\begin{aligned}
&\frac{\partial 
{\omega}}{\partial t}  = \nu\nabla^2{\omega} + \frac{\partial \mu}{\partial x}\frac{\partial \xi}{\partial y} - \frac{\partial \mu}{\partial y}\frac{\partial \xi}{\partial x} - \frac{\partial \phi}{\partial x}\frac{\partial \omega}{\partial y} + \frac{\partial \phi}{\partial y}\frac{\partial \omega}{\partial x} \\
&\frac{\partial {\mu}}{\partial t}  =  \eta\nabla^2{\mu} - \frac{\partial \phi}{\partial x}\frac{\partial \mu}{\partial y} + \frac{\partial \phi}{\partial y}\frac{\partial \mu}{\partial x} \\
&\nabla^2 {\phi} + {\omega} =  0 \\
&\nabla^2 {\mu} + \xi = 0 \\
\end{aligned},
\end{equation}
where $\omega$, $\mu$, $\phi$, $\xi$, $\nu$, and $\eta$ are the vorticity,  magnetic potential, streamfunction, current density, kinematic viscosity, and resistivity, respectively.  The first two PDEs in Eq. (\ref{eq:MHD}) are combined into a single system of nonlinear PDEs, where
\begin{equation}
\label{eq:combinedMHD}
    \frac{\partial}{\partial t}
    \begin{bmatrix}
    {\omega} \\
    {\mu}
    \end{bmatrix} =
    \begin{bmatrix}
    \nu \nabla^2 & {0} \\
    {0} & \eta \nabla^2 \\
    \end{bmatrix}
    \begin{bmatrix}
    {\omega} \\
    {\mu}
    \end{bmatrix} + 
    \begin{bmatrix}
    \frac{\partial {\mu}}{\partial x} \frac{\partial {\xi}}{\partial y} - \frac{\partial {\mu}}{\partial y} \frac{\partial {\xi}}{\partial x} - \frac{\partial {\phi}}{\partial x} \frac{\partial {\omega}}{\partial y} + \frac{\partial {\phi}}{\partial y} \frac{\partial {\omega}}{\partial x} \\
    \frac{\partial {\phi}}{\partial y} \frac{\partial {\mu}}{\partial x} - \frac{\partial {\phi}}{\partial x} \frac{\partial {\mu}}{\partial y}
    \end{bmatrix}.
\end{equation} 

In this example, $\omega$ and $\mu$ are computed at multiple points arranged into a square grid with a side length of 1 m.  The values of $\omega$ and $\mu$ in an $8 \times 8$ grid are obtained with VQHPM.  The initial vorticity and magnetic potential fields are both $\omega(x,y,t=0) = \mu(x,y,t=0) = 0.25\exp\left(-50 \left((x-0.5)^2 + (y-0.5)^2 \right)\right)$. The combined PDE system in Eq. (\ref{eq:combinedMHD}) is solved for 5 time steps which each occur for $\Delta t = 0.0002$ s. At each time step, the  homotopy series terms $\omega_v^{(j)}$'s and $\mu_v^{(j)}$'s are calculated by solving the initial linear deformation equation 
\begin{equation} \label{eq:zerothmhddeformationequation}
    \frac{\partial}{\partial t}
    \begin{bmatrix}
    {\omega}_v^{(0)} \\
    {\mu}_v^{(0)}
    \end{bmatrix} =
    \begin{bmatrix}
    \nu \nabla^2 & 0 \\
    0 & \eta \nabla^2 \\
    \end{bmatrix}
    \begin{bmatrix}
    {\omega}_v^{(0)} \\
    {\mu}_v^{(0)}
    \end{bmatrix}
\end{equation}
and $m$
higher-order linear deformation equations of the form
\begin{equation}
    \frac{\partial}{\partial t}
    \begin{bmatrix}
    {\omega}_v^{(j)} \\
    {\mu}_v^{(j)}
    \end{bmatrix} =
    \begin{bmatrix}
    \nu \nabla^2 & {0} \\
    {0} & \eta \nabla^2 \\
    \end{bmatrix}
    \begin{bmatrix}
    {\omega}_v^{(j)} \\
    {\mu}_v^{(j)}
    \end{bmatrix} + \sum_{k=1}^j
    \begin{bmatrix}
    \frac{\partial {\mu}_v^{(k-1)}}{\partial x} \frac{\partial {\xi}_v^{(j-k)}}{\partial y} - \frac{\partial {\mu}_v^{(k-1)}}{\partial y} \frac{\partial {\xi}_v^{(j-k)}}{\partial x} - \frac{\partial {\phi}_v^{(k-1)}}{\partial x} \frac{\partial {\omega}_v^{(j-k)}}{\partial y} + \frac{\partial {\phi}_v^{(k-1)}}{\partial y} \frac{\partial {\omega}_v^{(j-k)}}{\partial x} \\
    \frac{\partial {\phi}_v^{(k-1)}}{\partial y} \frac{\partial {\mu}_v^{(j-k)}}{\partial x} - \frac{\partial {\phi}_v^{(k-1)}}{\partial x} \frac{\partial {\mu}_v^{(j-k)}}{\partial y}
    \end{bmatrix}.
\end{equation}
The solutions ${\omega}_v^{(0)}(x,y,t=\Delta t)$ and ${\mu}_v^{(0)}(x,y,t=\Delta t)$ are first obtained by solving Eq. (\ref{eq:zerothmhddeformationequation}) with VQS, where ${\omega}_v^{(0)}(x,y,t=0)$ and ${\mu}_v^{(0)}(x,y,t=0)$ are set to the original initial conditions ${\omega}(x,y,t=0)$ and ${\mu}(x,y,t=0)$, respectively. Next,
the problem of solving the $j$\textsuperscript{th}-order deformation equation is decomposed into two sub-problems of solving the steady-state PDE
\begin{equation}
    \begin{bmatrix}
    \nu \nabla^2 & 0 \\
    0 & \eta \nabla^2 \\
    \end{bmatrix}
    \begin{bmatrix}
    \hat{{\omega}}^{(j)} \\
    \hat{{\mu}}^{(j)}
    \end{bmatrix} =
    \sum_{k=1}^j
    \begin{bmatrix}
    \frac{\partial \hat{{\phi}}^{(k-1)}}{\partial x} \frac{\partial \hat{{\omega}}^{(j-k)}}{\partial y} - \frac{\partial \hat{{\phi}}^{(k-1)}}{\partial y} \frac{\partial \hat{{\omega}}^{(j-k)}}{\partial x} - \frac{\partial \hat{{\mu}}^{(k-1)}}{\partial x} \frac{\partial \hat{{\xi}}^{(j-k)}}{\partial y} + \frac{\partial \hat{{\mu}}^{(k-1)}}{\partial y} \frac{\partial \hat{{\xi}}^{(j-k)}}{\partial x} \\
    \frac{\partial \hat{{\phi}}^{(k-1)}}{\partial x} \frac{\partial \hat{{\mu}}^{(j-k)}}{\partial y}
    \frac{\partial \hat{{\phi}}^{(k-1)}}{\partial x} \frac{\partial \hat{{\mu}}^{(j-k)}}{\partial y}
    \end{bmatrix}
\end{equation}
and the homogeneous time-dependent PDE
\begin{equation}
    \frac{\partial}{\partial t}
    \begin{bmatrix}
    \Tilde{{\omega}}^{(j)} \\
    \Tilde{{\mu}}^{(j)}
    \end{bmatrix} =
    \begin{bmatrix}
    \nu \nabla^2 & {0} \\
    {0} & \eta \nabla^2 \\
    \end{bmatrix}
    \begin{bmatrix}
    \Tilde{{\omega}}^{(j)} \\
    \Tilde{{\mu}}^{(j)}
    \end{bmatrix}.
\end{equation}
In the steady-state PDE, $\hat{{\phi}}$ and $\hat{{\xi}}$ are obtained by solving the third and fourth PDEs in Eq. (\ref{eq:MHD}) with FDM.  After the linear deformation equations are solved, the nonlinear solutions $\omega$ and $\mu$ on the $8 \times 8$ grid are approximated through Eq. (\ref{eq:homotopyseries}). The solution values on the $8 \times 8$ grid are subsequently used to interpolate $\omega$ and $\mu$ in a 100 $\times$ 100 grid through the Chebyshev spectral collocation method. 

\begin{figure}[h!]
        \centering
        \includegraphics[width=0.5\linewidth]{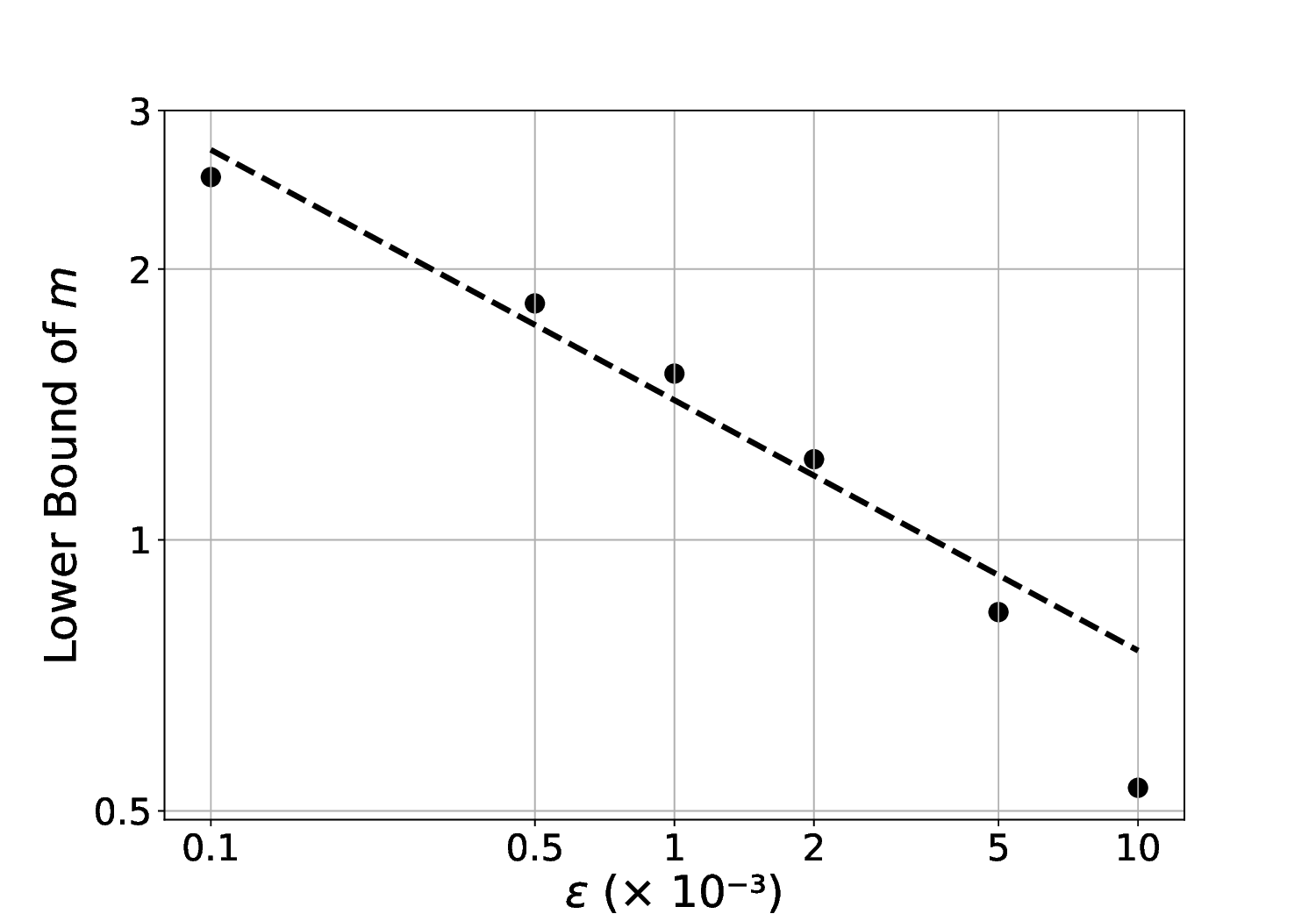}
        \caption{Lower bound of Corollary \ref{thm:truncation} vs. targeted approximation error of solving the reduced MHD equations}
        \label{fig:reducedMHD_truncationvserror}
\end{figure}

The targeted absolute differences for $\omega$ and $\mu$ are both selected as $\epsilon = 0.0125$.  This targeted error is 5\% of 0.25, which is the maximum absolute value of the initial $\omega$ and $\mu$ fields.  Assuming that $q = 0.1$, the lower bound of $m$ in Corollary \ref{thm:truncation} is 0.4335 as illustrated in Figure \ref{fig:reducedMHD_truncationvserror}.  Therefore, $m$ is set to 1.

\begin{figure}[h!]
    \centering
    \includegraphics[width=0.5\linewidth]{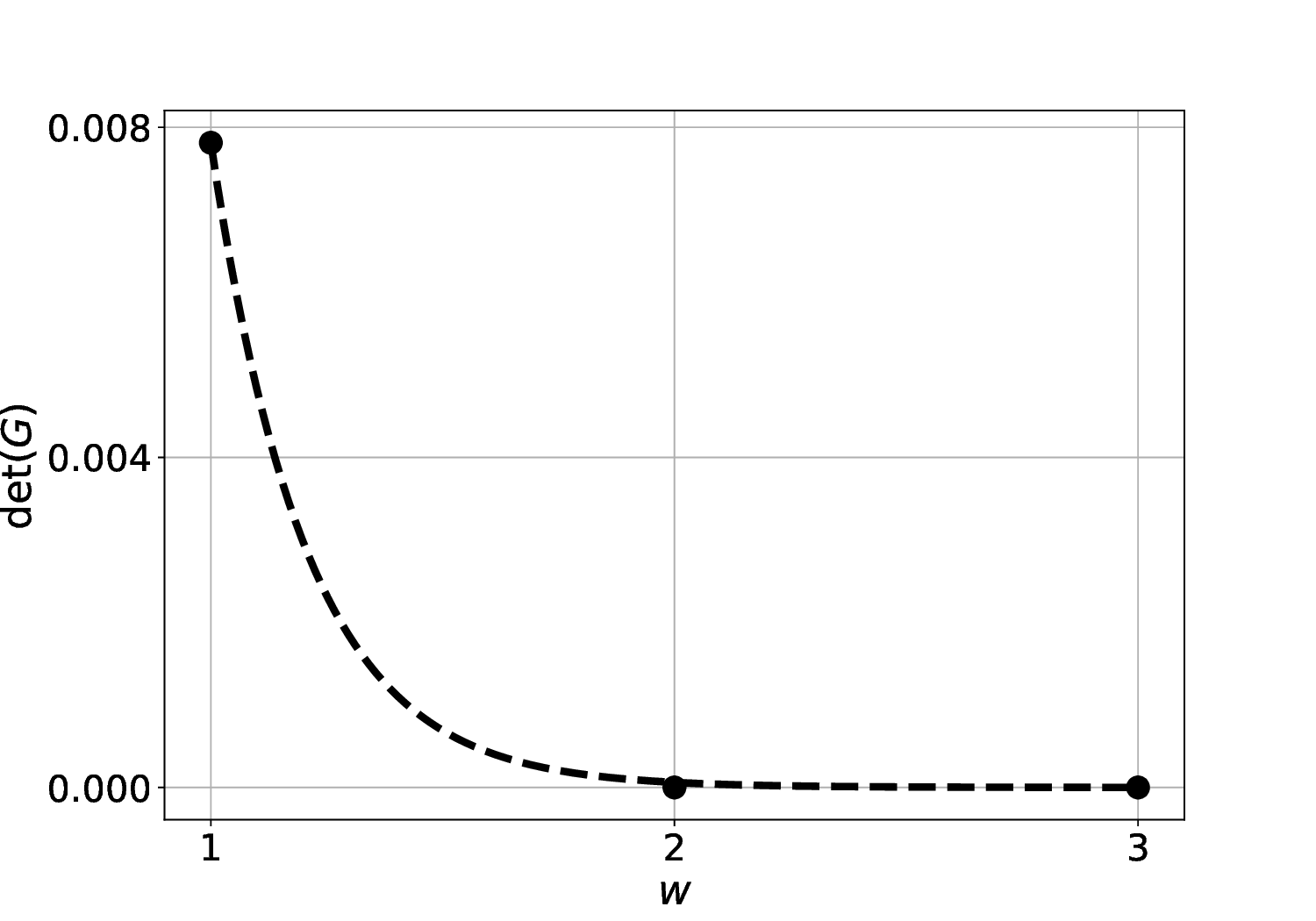}
    \caption{Estimated values of det($G$) at different circuit depths for the reduced MHD equations}
    \label{fig:reducedMHD_fsvsdepth}
\end{figure}

The estimated values of $\det(G)$ given different $w$ are shown as black dots in Figure \ref{fig:reducedMHD_fsvsdepth}. For each $w$ from 1 to 3, $\det(G)$ is estimated as the average of Eq. (\ref{eq:detg}) from 500 random samples of $\boldsymbol{\theta}$. As indicated by the dashed curve, $\det(G)$ decreases exponentially from 0.0078 to 0 as $w$ increases from 1 to 2. The number of redundant circuit parameters is minimized by setting $w$ to 1.

\begin{figure}[h!]
    \centering
    \begin{subfigure}{0.32\textwidth}
        \centering
        \includegraphics[width=\linewidth, trim={1.5cm 1.5cm 1.5cm 1.5cm}, clip]{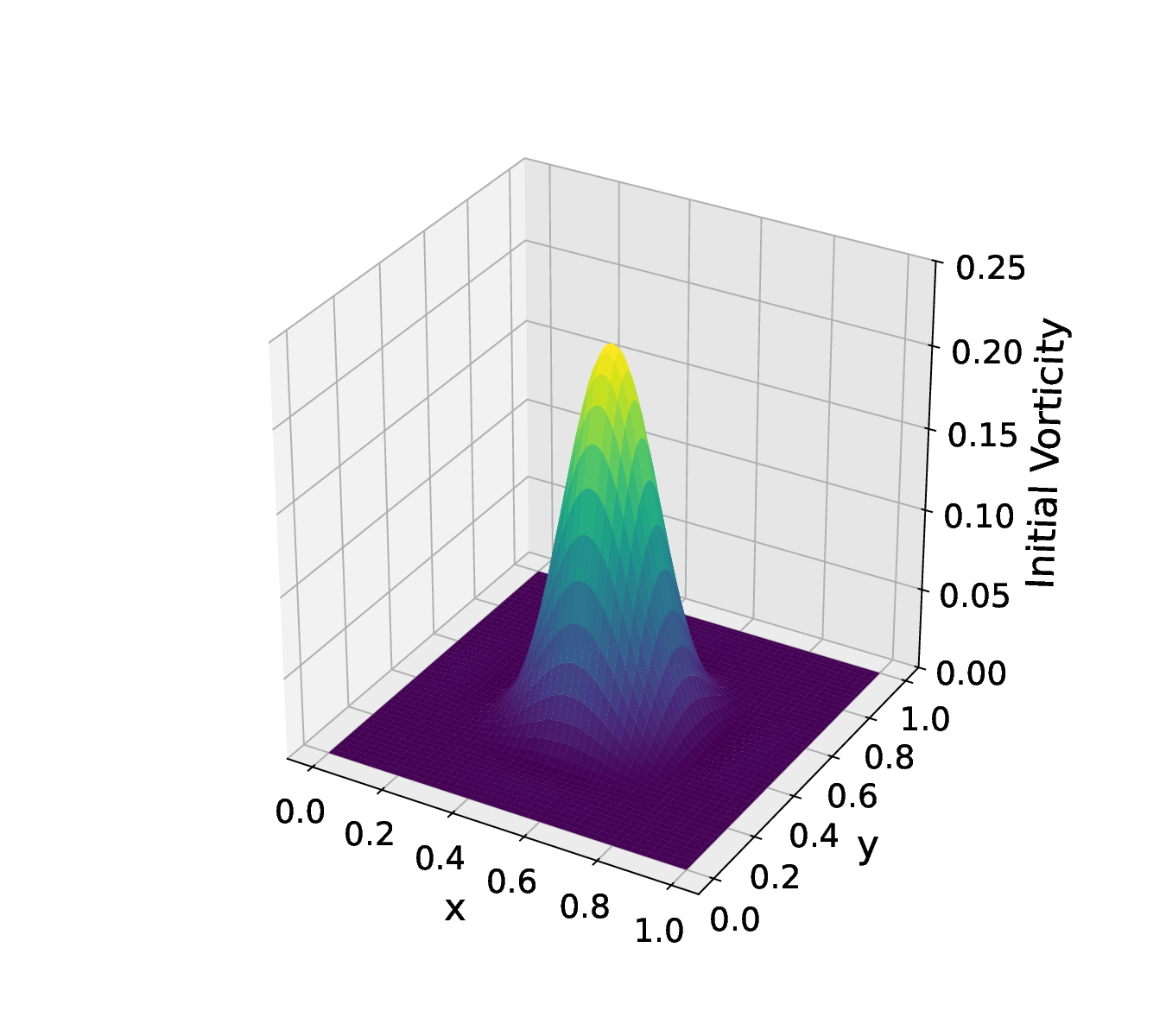}
        \caption{}
        \label{subfig:reducedMHD_initialvorticity}
    \end{subfigure}
    \begin{subfigure}{0.32\textwidth}
        \centering
        \includegraphics[width=\linewidth, trim={1.5cm 1.5cm 1.5cm 1.5cm}, clip]{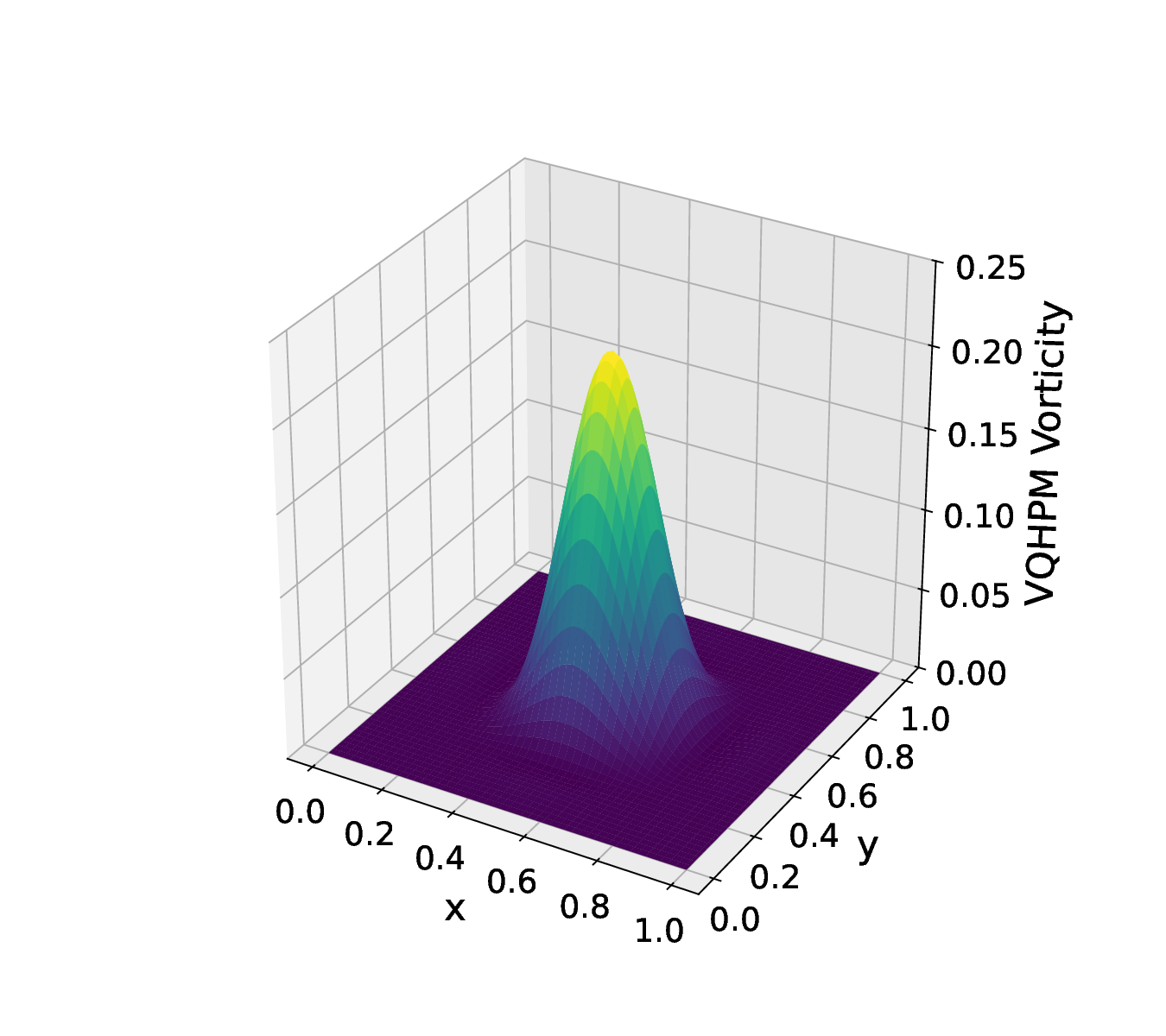}
        \caption{}
        \label{subfig:reducedMHD_finalvorticity_qhpm}
    \end{subfigure}
    \begin{subfigure}{0.32\textwidth}
        \centering
        \includegraphics[width=\linewidth, trim={1.5cm 1.5cm 1.5cm 1.5cm}, clip]{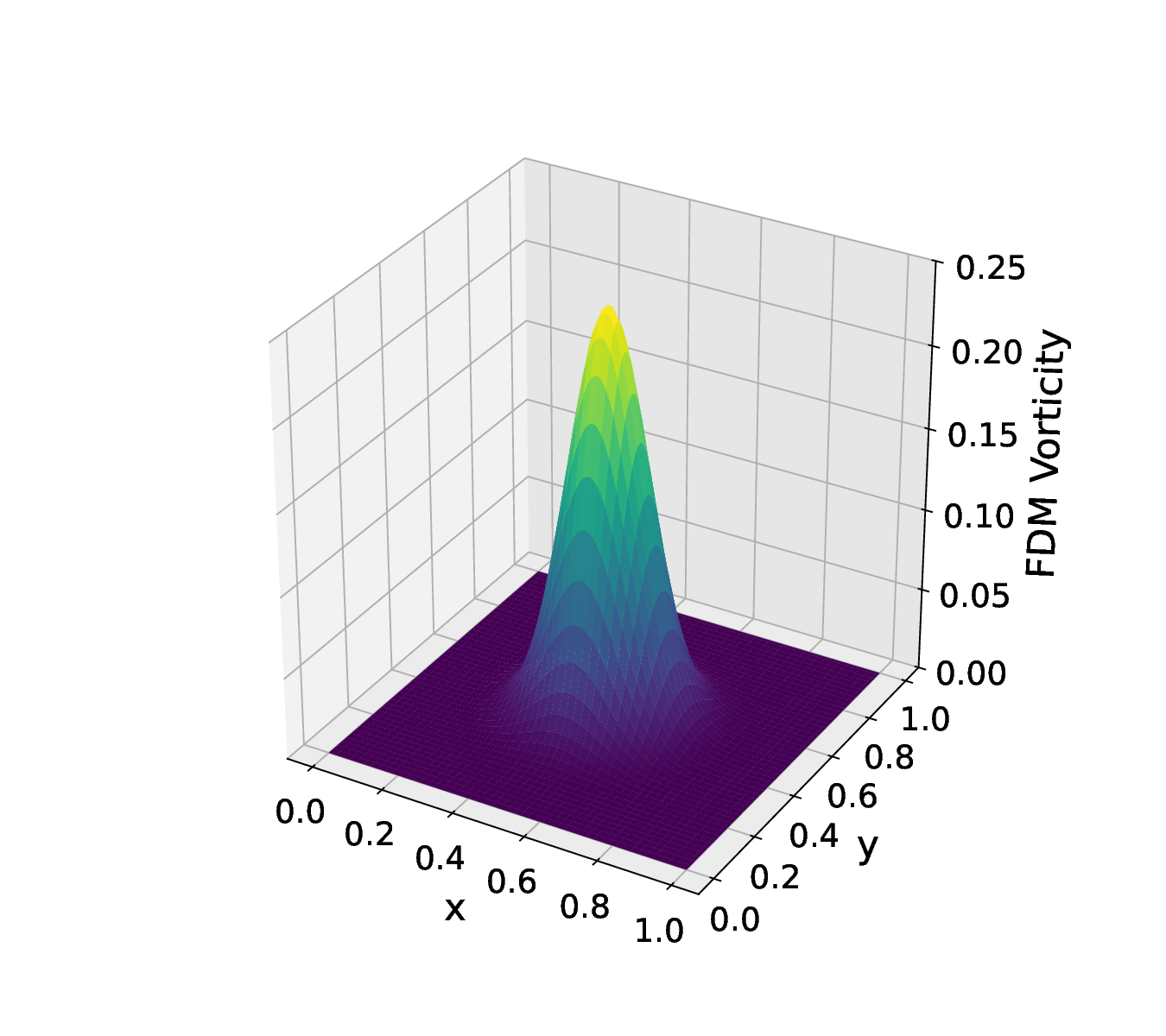}
        \caption{}
        \label{subfig:reducedMHD_finalvorticity_fdm}
    \end{subfigure}
    \begin{subfigure}{0.32\textwidth}
        \centering
        \includegraphics[width=\linewidth, trim={1.5cm 1.5cm 1.5cm 1.5cm}, clip]{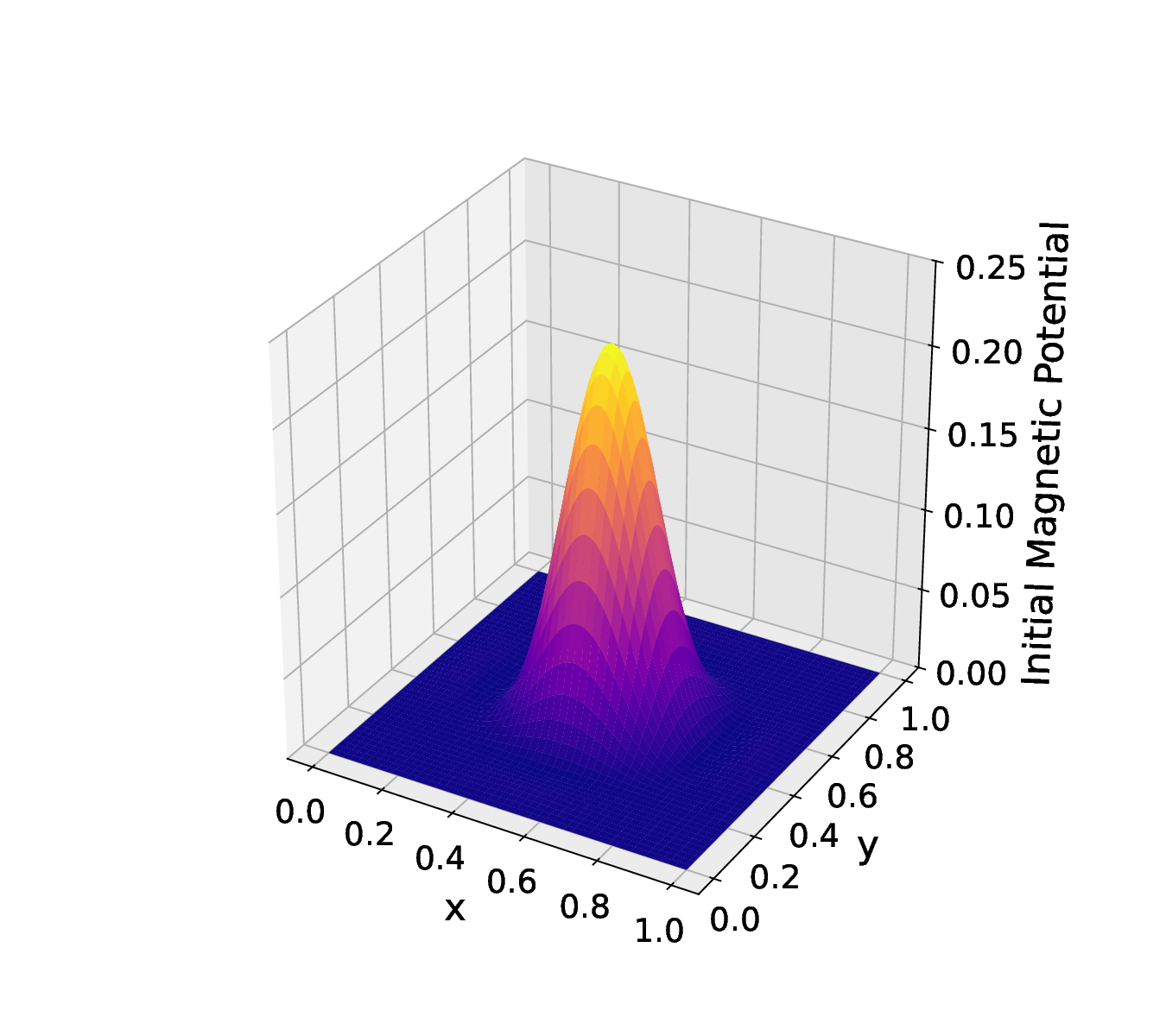}
        \caption{}
        \label{subfig:reducedMHD_initialmagpotential}
    \end{subfigure}
    \begin{subfigure}{0.32\textwidth}
        \centering
        \includegraphics[width=\linewidth, trim={1.5cm 1.5cm 1.5cm 1.5cm}, clip]{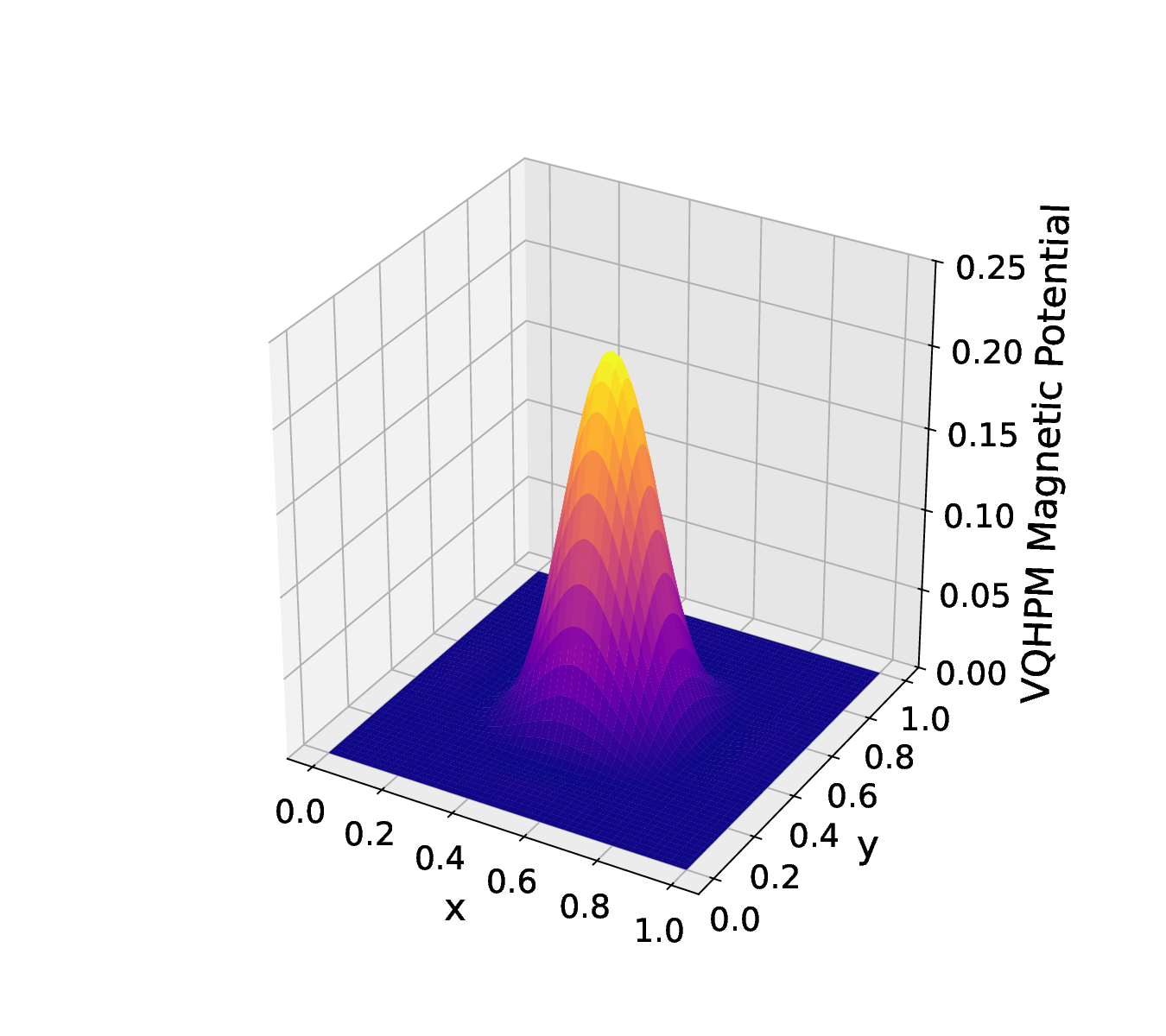}
        \caption{}
        \label{subfig:reducedMHD_finalmagpotential_qhpm}
    \end{subfigure}
    \begin{subfigure}{0.32\textwidth}
        \centering
        \includegraphics[width=\linewidth, trim={1.5cm 1.5cm 1.5cm 1.5cm}, clip]{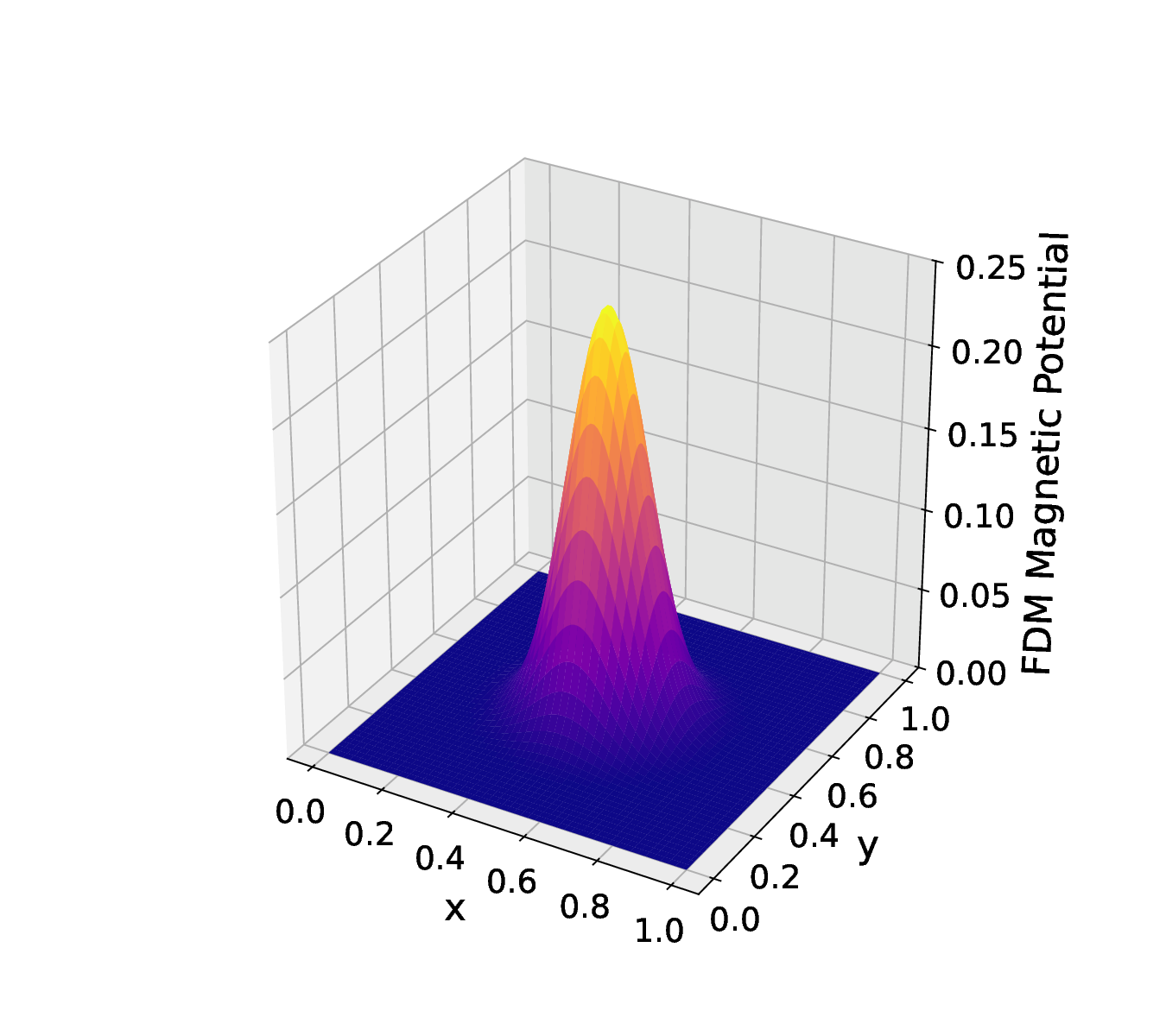}
        \caption{}
        \label{subfig:reducedMHD_finalmagpotential_fdm}
    \end{subfigure}
    \caption{Vorticity and streamfunction fields for the reduced MHD equations, including (a, d) initial condition at t = 0, (b, e) VQHPM results at t = 0.001 s, and (c, f) FDM results at t = 0.001 s}
    \label{fig:reducedMHD_finalfields}
\end{figure}

\begin{figure}[h!]
    \centering
    \begin{subfigure}{0.45\textwidth}
        \centering
        \includegraphics[width=\linewidth, trim={0.25cm 1.5cm 0.25cm 1.5cm}, clip]{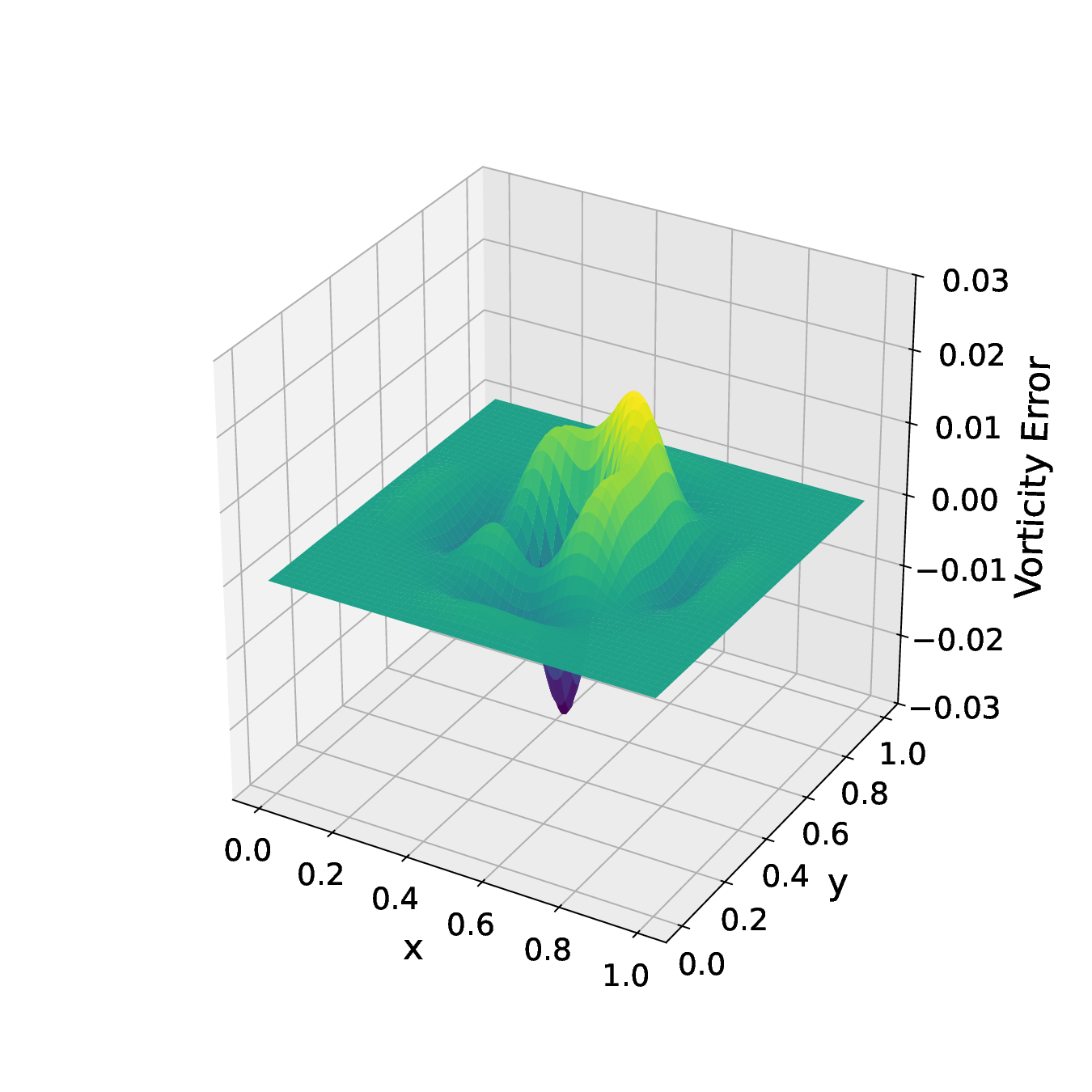}
        \caption{}
        \label{subfig:reducedMHD_vorticityerror}
    \end{subfigure}
    \begin{subfigure}{0.45\textwidth}
        \centering
        \includegraphics[width=\linewidth, trim={0.25cm 1.5cm 0.25cm 1.5cm}, clip]{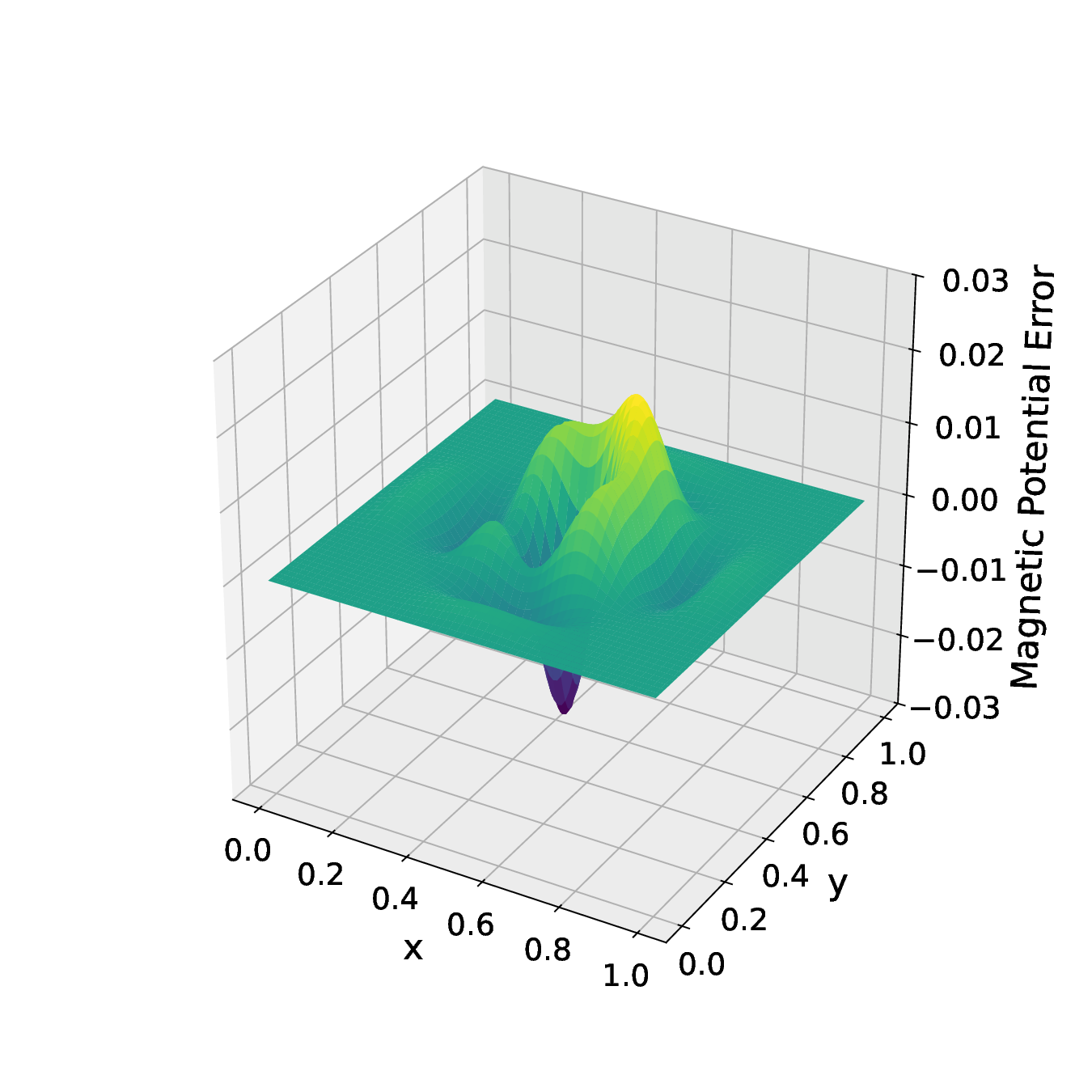}
        \caption{}
        \label{subfig:reducedMHD_magpotentialerror}
    \end{subfigure}
    \caption{(a) Vorticity and (b) magnetic potential errors of solving reduced MHD equations with VQHPM relative to FDM}
    \label{fig:reducedMHD_errorfields}
\end{figure}

\begin{figure}[h!] \label{fig:NSvorticityfields}
    \centering
    \begin{subfigure}{0.49\textwidth}
        \centering
        \includegraphics[width=\linewidth]{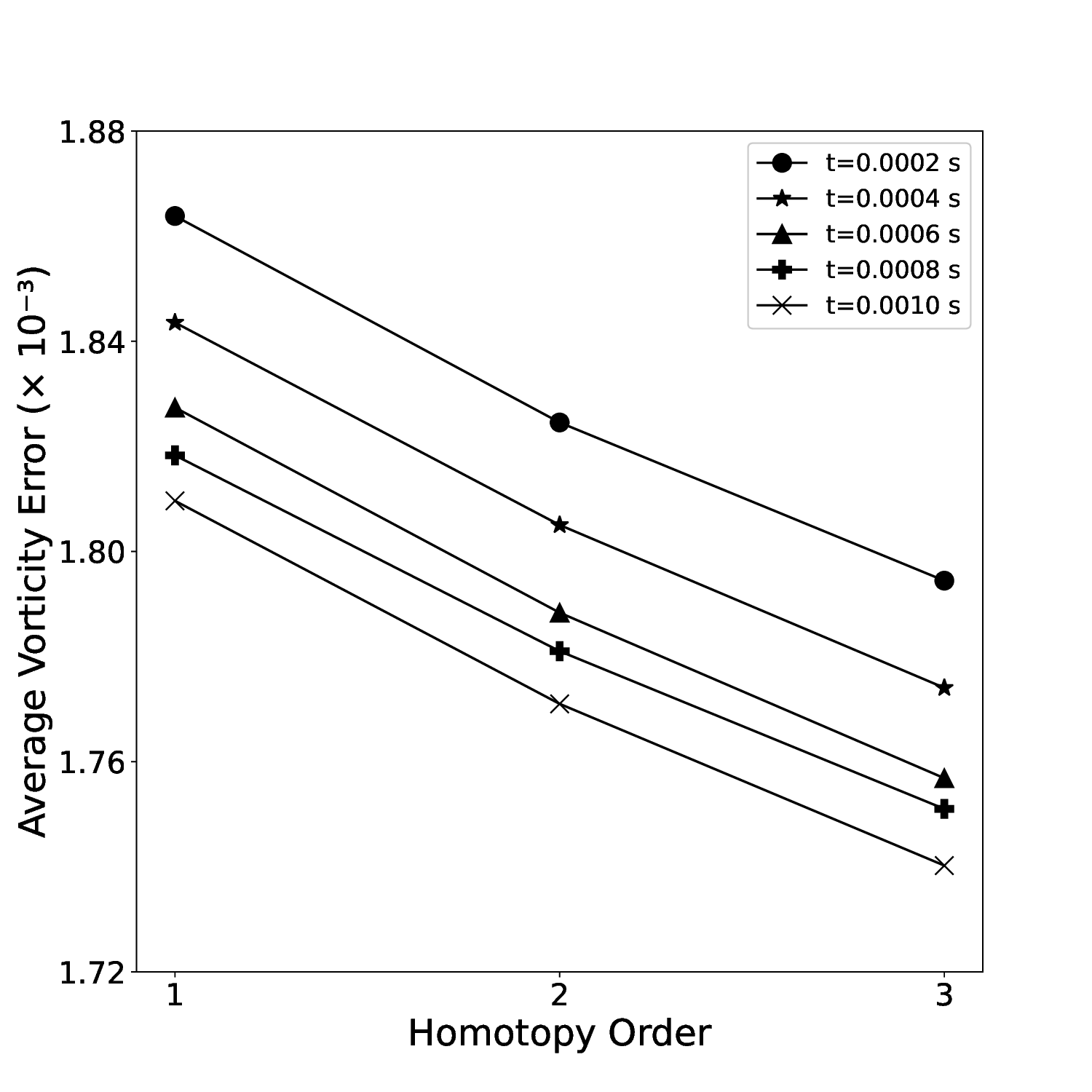}
        \caption{}
        \label{subfig:reducedMHD_averageabsolutevorticityerror}
    \end{subfigure}
    \begin{subfigure}{0.49\textwidth}
        \centering
        \includegraphics[width=\linewidth]{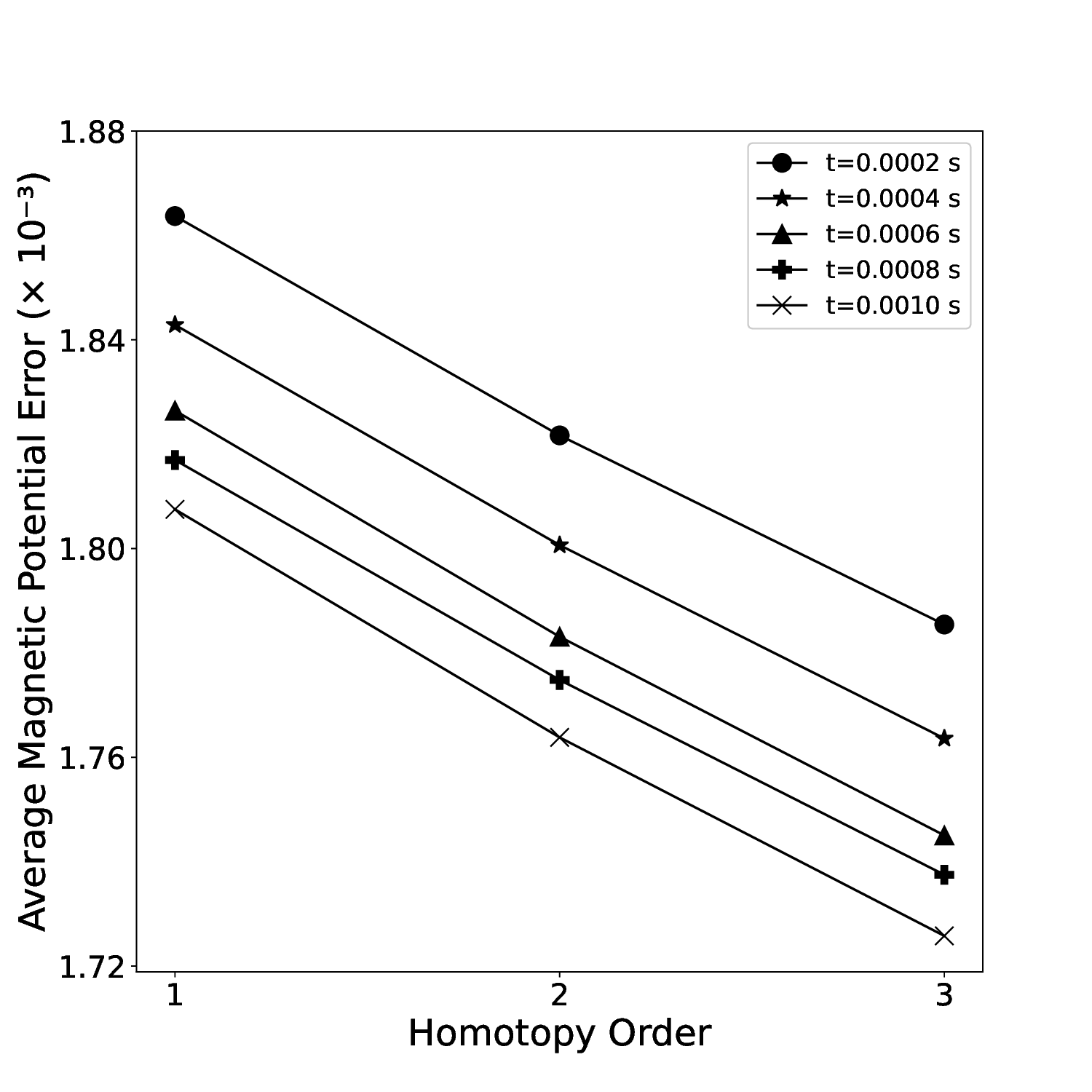}
        \caption{}
        \label{subfig:reducedMHD_averageabsolutemagpotentialerror}
    \end{subfigure}
    \caption{(a) Average vorticity and (b) magnetic potential errors of solving reduced MHD equations with VQHPM relative to FDM for different homotopy orders}
\end{figure}

Given the initial conditions for $\omega$ and $\mu$ in Figures \ref{subfig:reducedMHD_initialvorticity} and \ref{subfig:reducedMHD_initialmagpotential}, it is observed that the final solution fields at $t = 0.001$ s are highly similar between VQHPM and FDM.  The values of $\omega$ and $\mu$ obtained from VQHPM in Figures \ref{subfig:reducedMHD_finalvorticity_qhpm} and \ref{subfig:reducedMHD_finalmagpotential_qhpm} closely approximate the respective solution fields in Figures \ref{subfig:reducedMHD_finalvorticity_fdm} and \ref{subfig:reducedMHD_finalmagpotential_fdm}.  According to Figures \ref{subfig:reducedMHD_vorticityerror} and \ref{subfig:reducedMHD_magpotentialerror}, the absolute vorticity and magnetic potential differences between VQHPM and FDM at most grid points are less than 0.0125, which is the selected value of $\epsilon$. Out of all 10,000 grid points, 271 of them involve absolute differences which are larger than $\epsilon$, where the maximum absolute difference is 0.02804.

The convergence behavior of VQHPM with respect to $m$ is also analyzed.  For each $t$, it is observed in Figures \ref{subfig:reducedMHD_averageabsolutevorticityerror} and \ref{subfig:reducedMHD_averageabsolutemagpotentialerror} that the average absolute errors for $\omega$ and $\mu$ decrease as $m$ increases from 1 to 2. This means that a larger value of $m$ results in the homotopy series more closely approximating the nonlinear solutions.

\section{Discussions and Conclusions} \label{sec:conclusions}
In this paper, VQHPM is proposed as a new method to solve nonlinear PDEs on current quantum computers. The proposed VQHPM improves the scalability of solving nonlinear PDEs by converting the original nonlinear problem to linear deformation equations with the homotopy perturbation method. The dimension of the Hilbert space remains unchanged during the linearization process.  The VQHPM also improves scalability through which the linear deformation equations are solved with a VQS framework. The number of qubits can be decreased by the functional expansion strategy, and the parameterized circuit depth is reduced by utilizing a hardware-efficient ansatz.  For the vorticity transport and reduced MHD equations, it is demonstrated that VQHPM results in solutions which closely approximate the solutions obtained from FDM.  This occurs when the homotopy order and VQS circuit depth are set to minimal values.  Furthermore, it is observed that the homotopy series converge towards nonlinear solutions. As the homotopy order increases, the average approximation errors between the VQHPM and FDM solutions decrease.

Although the approximated solutions to the vorticity transport and reduced MHD equations are fairly accurate, the approximation error of low-order homotopy series will be significantly increased for highly nonlinear PDEs.  In highly nonlinear problems, the magnitudes of nonlinear correction terms are large relative to the initial guess term.  This results in homotopy series failing to converge towards nonlinear solutions. As future work, VQHPM will be further evaluated for highly nonlinear PDEs with higher order homotopy orders.

The ability of VQHPM to converge towards nonlinear PDE solutions also depends on discretization error.  Discretization error arises after the spatial domain of a solution field is discretized into a grid. Because the solutions to all linear deformation equations are discretized approximations of continuous solution fields, discretization error is propagated throughout the homotopy series. Future work will also involve analyzing the effect of discretization error on the convergence behavior of the homotopy series.

In both examples of this paper, the parametrized quantum circuit is sufficient for finding solutions to the linear deformation equations. However, solutions to several other nonlinear PDEs cannot be obtained with the same circuit architecture when the solutions do not overlap with the variational manifold of the circuit. Further efforts will be made to generalize VQHPM for other nonlinear PDEs by incorporating an adaptive circuit architecture protocol similar to variational quantum algorithms such as ADAPT-VQE \cite{grimsley2019adaptive} and ADAPT-QAOA \cite{zhu2022adaptive}. The architecture of the quantum circuit architecture will be modified between time steps in order to increase the extent of Hilbert space exploration.

Future work will also focus on reducing the computational runtime of VQHPM. This will be achieved in at least two ways. First, the VQS framework will be modified to directly solve nonhomogeneous linear differential equations.  This eliminates the need to decompose each linear deformation equation into two homogeneous linear PDEs.  As a result, the number of homotopy series terms is reduced from $2m+1$ to $m+1$.  Second, the time step size will be significantly increased so that fewer linear deformation equations are solved. The circuit parameter update rules based on Euler's method must be modified so that the evolutions of parameters do not deviate from the correct time evolutions.

\bibliographystyle{unsrt}

\bibliography{references}

\end{document}